\documentclass[twocolumn,prl,aps,superscriptaddress,10pt,longbibliography]{revtex4-2}
\usepackage{mathtools}
\usepackage{amsbsy}
\usepackage{amstext}
\usepackage[unicode=true,pdfusetitle,
 bookmarks=false,
 breaklinks=false,pdfborder={0 0 1},backref=false,colorlinks=false]
 {hyperref}
\usepackage[dvipsnames]{xcolor}

\hypersetup{colorlinks=true,citecolor=RoyalPurple,filecolor=black,linkcolor=RoyalPurple,urlcolor=RoyalPurple}

\makeatletter

\@ifundefined{textcolor}{}
{%
 \definecolor{BLACK}{gray}{0}
 \definecolor{WHITE}{gray}{1}
 \definecolor{RED}{rgb}{1,0,0}
 \definecolor{GREEN}{rgb}{0,1,0}
 \definecolor{BLUE}{rgb}{0,0,1}
 \definecolor{CYAN}{cmyk}{1,0,0,0}
 \definecolor{MAGENTA}{cmyk}{0,1,0,0}
 \definecolor{YELLOW}{cmyk}{0,0,1,0}
}
\newcommand\ket[1]{\left|#1\right\rangle}

\usepackage{txfonts, color}
\usepackage{dsfont}

\usepackage{graphicx}
\usepackage{float}

\usepackage{amsmath, amssymb, amsfonts, amsthm, bbm}
\usepackage{physics}
\usepackage[normalem]{ulem}

\usepackage{centernot}
\newcommand{\notimplies}{\centernot\implies}

\newtheorem{theorem}{Theorem}

\newtheorem{corollary}{Corollary}[theorem]
\newtheorem{observation}{Observation}
\makeatother

\begin{document}

\title{Marginal second order moments do not suffice for entanglement detection}
\author{A. Garcia-Velo}
\affiliation{IMDEA Networks, Avda. Mar Mediterr{\'a}neo, 22, 28918 Legan{\'e}s, Spain}
\author{M. Paraschiv}
\affiliation{IMDEA Networks, Avda. Mar Mediterr{\'a}neo, 22, 28918 Legan{\'e}s, Spain}
\author{Y. Ban}
\affiliation{Instituto de Ciencia de Materiales de Madrid (CSIC), Cantoblanco, 28049 Madrid, Spain}
\author{E. Torrontegui}
\email{eriktorrontegui@gmail.com}
\affiliation{Departamento de F{\'i}sica, Universidad Carlos III de Madrid, Avda. de la Universidad 30, 28911 Legan{\'e}s, Spain}
\author{R. Puebla}
\affiliation{Departamento de F{\'i}sica, Universidad Carlos III de Madrid, Avda. de la Universidad 30, 28911 Legan{\'e}s, Spain}

\begin{abstract}
The complete knowledge of the global and marginal second-order moments of a quantum state is in general insufficient for entanglement detection. By deriving the conditions on local-unitary (LU) transformations through the second-order moments, we construct pairs of separable and entangled states that are not LU equivalent, contain different amount of entanglement, and are not  equivalent under stochastic local operations and classical communication, even though the states share identical global and marginal second-order moments. 
\end{abstract}

\maketitle

{\em Introduction.--} Entanglement is arguably the most prominent hallmark of quantum mechanics~\cite{Horodecki_2009}, with profound implications in our understanding of physical reality~\cite{Einstein_35, Bell_64, Brunner_2014}. Beyond its foundational significance, entanglement serves as a key resource in the development of quantum technologies~\cite{Dowling_2003}. It plays a central role in a wide range of applications, from quantum information processing and computation~\cite{Nielsen_2010} to quantum metrology~\cite{Paris_09, Giovannetti_2011}.

The quantification and detection of entanglement is therefore of paramount importance~\cite{Friis_2019}, not only from a fundamental point of view, but also for the correct calibration of physical systems and the determination of quantum correlations beyond a classical setting. Traditional entanglement characterization protocols, however, face two main problems. On the one hand, if no prior knowledge of the studied state is known, one has to resort to methods based on standard quantum state tomography. These demand a number of measurements that grows exponentially with the number of particles~\cite{Blume-Kohout_2010, Gross_2010}, although prior knowledge of the state may improve its efficiency~\cite{Somma_2006, Guhne_2007, Guhne_2009, Cramer_2010, Bowles_2020}. On the other hand, such protocols need a careful alignment between local reference frames, which may be either difficult or impossible to implement~\cite{Bartlett_2007}.

Over the past decades, novel approaches have been developed to address the alignment problem~\cite{Walter_2013, Aguilar_2015, Aolita_2007, D'Ambrosio_2012}. Of special interest are reference-frame independent entanglement criteria, which leverage the invariance of quantum correlations under local reference frame transformations~\cite{Aschauer_2004, Vicente_2007, Vicente_2008, Badziag_2008, Hassan_2008, Vicente_2011, Laskowski_2011, Gabriel_2013}. Such methods have proven valuable in testing Bell-type inequalities~\cite{Liang_2010, Wallman_2011, Shadbolt_2012, Wallman_2012, Shadbolt_2012_2, Palsson_2012, Senel_2015, Rosier_2017, Fonseca_2018}, quantum key distribution~\cite{Laing_2010, Wabnig_2013, Slater_2014} and Schmidt decomposition estimation \cite{Laskowski_2012, Laskowski_2013}. Although they do not require local reference frame alignment, some control over the choice of local reference frames is typically still necessary.


In this regard, methods based on random measurements for entanglement detection hold great potential. They do not requiring control or knowledge over local reference frame alignment, even under limited measurement control and with a reduced number of measurements~\cite{Knips_2020, Wyderka_2023, Cieslinski_2024}. Indeed, random sequences and/or measurements have emerged as a valuable method in other scenarios, e.g. for probing novel quantum many-body physics~\cite{Skinner_2019, Noel_2022, Koh_2023, Morvan_2024} and states~\cite{Elben_2019, Elben_2020, Elben_2020_2, Notarnicola_2023}, for quantum device verification and characterization~\cite{Knill_2008, Elben_2020_3, Helsen_2022}, or for the estimation of relevant functions of the state~\cite{Elben_2018, Brydges_2019, Huang_2020}.

In particular, entanglement criteria have been proposed in terms of the statistical moments of the probability distributions obtained after random measurements~\cite{Knips_2020_2, Cieslinski_2024}. These methods use either global~\cite{Tran_2015, Tran_2016} or marginal~\cite{Knips_2020, Shravan_2024} 2nd-order moments, or combines global 2nd- and 4th-order moments~\cite{Ketterer_2020, Ketterer_2018}, which are experimentally accessible quantities in practical quantum platforms. However, it remains as an open question to what extent these moment-based methods can help in detecting entanglement in a fully general setting, that is, for mixed states in systems of arbitrary dimension.

In this Letter, we present results that constitute a significant step in this direction. We investigate whether marginal 2nd-order moments are sufficient to characterize local-unitary (LU) equivalent quantum states. Specifically, we ask: Given two states $\varrho$ and $\tilde{\varrho}$ of the same physical system, does the equality of their global and marginal 2nd-order moments imply that the states $\varrho$ and $\tilde{\varrho}$ are related by a LU transformation? We demonstrate that the answer is negative in general, with the equivalence holding only for pure states of two qubits. Thereby showing that merely information about marginal and global 2nd-order moments do not suffice, in general, for entanglement detection, let alone classification or quantification. Notably, we provide examples of separable and entangled states with identical 2nd-order moments, which not only proves the question under study, but also that two states with equal global and marginal 2nd-order moments may belong to different stochastic local operations and classical communication (SLOCC) classes.

To prove this result, we first introduce the notation for state representation, random measurements and 2nd-order moments. We then derive the conditions that two LU-equivalent states need to fulfill in terms of their Bloch vectors, and how they manifests in their global and marginal 2nd-order moments. Based on this analysis, we construct a counterexample where a separable and an entangle state share identical 2nd-order moments. Note, however, that in the specific case of two-qubit pure states identical marginal moments is a necessary and sufficient condition for LU-invariance~\cite{Tran_2015, Tran_2016}.

{\em State representation.--} Let $\mathcal{H}$ be the Hilbert space of a system with dimension $\dim \mathcal{H} = d$. Furthermore, let $\mathbb{H}(\mathcal{H})$ be the associated real Hilbert space of all bounded Hermitian operators acting on $\mathcal{H}$ with respect to the Hilbert-Schmidt inner product. This defines the set of all possible quantum observables, and has $\dim \mathbb{H}(\mathcal{H}) = d^2 - 1$. Quantum states form the subset $\mathcal{D(H)} = \qty{\varrho \in \mathbb{H}(\mathcal{H}) | \tr \varrho = 1, \varrho \geq 1}$. We say the state is pure if $\varrho \in \mathcal{D(H)}$ is rank-1, such that $\varrho = \dyad{\psi}$. Any function of the state $\varrho \in \mathcal{D(H)}$, may be decomposed in terms of measurements of certain observables $A \in \mathbb{H}(\mathcal{H})$, and estimation of its expected values $\expval{A, \varrho} = \tr A \varrho$. \cite{Wyderka_2023}.

The Gell-Mann operators, $\qty{\lambda_i}_{i = 0}^{d^2- 1}$, whose elements satisfy $\langle\lambda_i, \lambda_j\rangle = d \delta_{ij}$ and $\tr \lambda_i = d \delta_{i0}$, constitute a suitable basis for $\mathbb{H}(\mathcal{H})$. Thus, they allow to write any $\varrho \in \mathcal{D(H)}$ in terms of the expectation values $r_i = \expval{\varrho, \lambda_i}$. If $\varrho \in \mathcal{D(H)}$, then $\varrho = \frac{1}{d} \qty(\mathbbm{1} + \sum_{i = 1}^{d^2 - 1} r_i \lambda_i)$, with $r_i \in \mathbb{R}$, $\forall i \in \qty{1, \ldots, d^2 - 1}$. Defining the vector $\vec{r} = \qty(r_i \big| i \in \qty{1, \ldots, d^2 - 1} )$, then $\norm{\vec{r}}^2 \leq d - 1$, with $\norm{\vec{r}}^2 = d - 1$ if and only if $\varrho$ is pure.

Let our system be composed of $N$ subsystems, and denote $\hat{N} \equiv \qty{1, 2, \ldots, N}$ the set of the first $N$ integers. Each subsystem has an associated Hilbert space $\mathcal{H}^{(n)}$ with respective dimensions $\dim \mathcal{H}^{(n)} = d^{(n)}$ such that $\mathcal{H} = \bigotimes_{n \in \hat{N}} \mathcal{H}^{(n)}$.

Explicitly acknowledging the tensor product structure, the state becomes $\varrho = \frac{1}{d} \sum_{i^{(1)} = 0}^{d^{(1)2} - 1} \cdots \sum_{i^{(N)} = 0}^{d^{(N)2} - 1} r_{i^{(1)}, \ldots, i^{(N)}} \bigotimes_{n \in \hat{N}} \lambda_{i^{(n)}} \equiv \frac{1}{d} \sum_{i^{(\hat{N})} = 0}^{d^{(n)2} - 1} r_{i^{(\hat{N})}} \lambda_{i^{(\hat{N})}}$. First, note that $i^{(\hat{N})}$ denotes the tuple of indices $(i^{(1)}, \ldots, i^{(N)}) = (i^{(n)} | n \in \hat{N})$, where every $i^{(n)}$ corresponds to the indexing of the local Hilbert space $\mathcal{H}^{(n)}$. Second, $\lambda_{i^{(\hat{N})}}$ is just a notation for $\bigotimes_{n \in \hat{N}} \lambda_{i^{(n)}}$. As a general rule, we will use a set as upper index to denote a tuple of indices. Note the difference between $i^{(N)}$ and $i^{(\hat{N})}$. Finally, $r_{0, \ldots, 0} = 1$. 



In this manner, one can explicitly distinguish the elements of $\vec{r}$ that define each subsystem and write
\begin{equation}
    \varrho = \frac{1}{d} \qty( \mathbbm{1}_{\mathcal{H}} + \sum_{\hat{M} \subseteq \hat{N}} \sum_{i^{(\hat{M})} = 1}^{d^{(m)2} - 1} r_{i^{(\hat{M})}} \lambda_{i^{(\hat{M})}}) \, ,
\end{equation}
where $\sum_{\hat{M} \subseteq \hat{N}}$ denotes the sum over all possible subsets of $\hat{N}$. Following the previous notation, let $i^{(\hat{M})}$ be the tuple of indices that contain the ones corresponding to the local Hilbert spaces included in partition $\hat{M}$, that is $i^{(\hat{M})} \equiv (i^{(m)} | m \in \hat{M})$. Equivalently, the same notation may be used to denote the tuple of indices $i^{(\hat{N})}$ with zeros in the positions not contained in $\hat{M}$.

{\em Random measurements.--} The random measurements approach assumes independent measurements in each individual subsystem, i.e. $A = \bigotimes_{n \in \hat{N}} A^{(n)}$; with $A^{(n)} \in \mathbb{H}\qty[\mathcal{H}^{(n)}]$, $\forall n \in \hat{N}$. Each observable $A^{(n)}$ is randomly sampled with respect to the uniform measure in $\mathbb{H}\qty[\mathcal{H}^{(n)}]$. This is equivalent to initially sampling elements from the special unitary group, $U^{(n)} \in \mathcal{SU}\qty[d^{(n)}]$ with respect to the Haar measure, $\forall n \in \hat{N}$. Then applying an LU transformation, $\varrho \to U^\dagger \varrho U$, where $U = \bigotimes_{n \in \hat{N}} U^{(n)}$. And finally measuring with a fixed product observable. This results in a probability distribution for $\langle \varrho, U A U^\dagger\rangle$, whose behavior depends on the correlations present in the state~\cite{Knips_2020, Knips_2020_2}. This method can be applied even under strong local unitary noise, assuming both high-rate state generation (compared to the time scales of the unitary noise in the quantum channel) and uniform sampling of local observables~\cite{Knips_2020}.



One can define the $t$-th order moments of the probability distribution of the random correlation function for the local observable $A$ as
\begin{equation} \label{Equation: t-th order moments}
    \mathcal{R}^{(t)}_A (\varrho) \equiv C \int_{\mathcal{SU}\qty[d^{(1)}]} \dd U^{(1)} \cdots \int_{\mathcal{SU}\qty[d^{(N)}]} \dd U^{(N)} \expval{\varrho, U A U^\dagger}^t \, ,
\end{equation}
where $C$ is a normalization constant, and $\int_{\mathcal{SU}\qty[d^{(n)}]} \dd U^{(n)}$ denotes the Haar unitary integral over $\mathcal{SU}[d^{(n)}]$~\cite{Ketterer_2018, Ketterer_2020, Knips_2020_2, Cieslinski_2024}. In this particular problem, knowledge of all the moments allows to uniquely determine the distribution \cite{Knips_2020_2}. Note that all the moments for odd $t$ vanish due to the symmetry of the probability distribution.

Unitary $t$-designs can be used to simplify Eq.~\eqref{Equation: t-th order moments}. They are a set of unitary operators $\qty{U_k \in \mathcal{SU}(d)}_{k = 1}^{K^{(t)}}$ such that $\frac{1}{K^{(t)}} \sum_{k = 1}^{K^{(t)}} P_t(U_k) = \int_{S_u[d]} P_t(U) \dd U$, for all homogeneous polynomials $P_t$ on $\mathcal{SU}(d)$ of degree smaller or equal than $t$~\cite{Ketterer_2018, Ketterer_2020}. In particular, $2$nd- and $4$th-order moments have been obtained relying on unitary $3$-designs~\cite{Webb_2016, Zhu_2016} and unitary $5$-designs~\cite{Gross_2007}, respectively~\cite{Ketterer_2018, Ketterer_2020}. In the particular case of local dimension $2$ for every subsystem (qubits), the two-to-one homomorphism between $\mathcal{SU}(2)$ and $\mathcal{SO}(3)$ allows to simplify the Haar integration over the hole unitary group as an integral over the uniform measure in the $3$-dimensional sphere. Quantum designs also allow for the computation of $t$-th order moments with a reduced number of fixed measurements, instead of random ones. This, however, requires control on local measurement directions but still without alignment between subsystems.


{\em Global second order moments.--} The 2nd-order moment has been extensively studied for entanglement detection~\cite{Cieslinski_2024}. It can be written as a function of the elements in $\vec{r}$, a simple relation that cannot be generalized to higher order moments~\cite{Cieslinski_2024, Tran_2015, Tran_2016}:
\begin{equation} \label{Equation: 2nd-order moment}
    \mathcal{R}^{(2)} (\varrho) = \sum_{i^{(\hat{N})} = 1}^{d^{(n)2} - 1} r_{i^{(\hat{N})}}^2 = \norm{\vec{r}^{(\hat{N})}}^2 \, ,
\end{equation}
upon a suitable choice of $C$ in Eq.~\eqref{Equation: t-th order moments}. For pure states of any local dimension $d^{(n)}$ and number of subsystems $N$, lower bounds have been derived, which saturates if and only if the state is separable. Thus $\mathcal{R}^{(2)}$ alone is enough to detect entanglement~\cite{Tran_2015, Tran_2016}. For mixed states, however, the bound is no longer satisfied and $\mathcal{R}^{(2)}$ has not such a direct implication in entanglement detection~\cite{Tran_2016}. Nevertheless, a necessary and sufficient condition has been found for rank-2 states, by convex roof construction~\cite{Tran_2016}. In addition, necessary criteria in terms of $\mathcal{R}^{(2)}$ can be formulated for the classification into stochastic local operations and classical communication (SLOCC) entanglement classes~\cite{Ketterer_2020}.

We denote that $\varrho, \tilde{\varrho} \in \mathcal{D(H)}$ are LU-equivalent as $\varrho \sim \tilde{\varrho}$. If $\varrho \sim \tilde{\varrho}$, then $\mathcal{R}^{(t)} (\varrho) = \mathcal{R}^{(t)} (\tilde{\varrho})$. This directly follows from the definition of the Haar measure in Eq.~\eqref{Equation: t-th order moments}~\cite{Collins_2006, Puchala_2017, Mele_2024}. Thus, $\mathcal{R}^{(t)} (\varrho)$ is invariant under LU transformations. Nonetheless, there exist states $\varrho, \tilde{\varrho} \in \mathcal{D(H)}$ such that $\mathcal{R}^{(2)} (\varrho) = \mathcal{R}^{(2)} (\tilde{\varrho})$, and yet $\varrho \nsim \tilde{\varrho}$. An example proving this statement is given in the End Matter (cf.~\ref{Proof: Same moment different classes}). This means that $\mathcal{R}^{(2)}$ does not distinguish between LU-entanglement classes. In general, SLOCC classes are not distinguished either, even for two qubits. 


The global 2nd-order moment provides valuable information, but alone is not sufficient to unequivocally detect entanglement. Two distinct paths are taken to improve this moments-based method. On the one hand, considering also higher-order moments to gain further insight into the global correlation distribution. Indeed, the use of $\mathcal{R}^{(2)}$ and $\mathcal{R}^{(4)}$ for qubits has proven useful, allowing for several necessary or sufficient conditions in entanglement detection~\cite{Ketterer_2018} or for SLOCC classification~\cite{Ketterer_2020}. Yet, its extension to arbitrary dimensions and/or to even higher-order moments is challenging. On the other hand, one may access the 2nd-order moments of the global and all possible reduced states, the so-called marginals 2nd-order moments. These marginals aid in identifying correlations between different subsets of particles, without compromising its generality and computational complexity. This latter approach is the focus of our work, and as we show in the following it fails to be sufficient for entanglement detection.

{\em Marginal second order moments.--} Let $\varrho \in \mathcal{D(H)}$ and $\hat{M} \subseteq \hat{N}$. The reduced state of $\varrho$ restricted to subsystem $\hat{M}$ is given by $\tr_{\hat{N} \setminus \hat{M}} \varrho \equiv \varrho^{(\hat{M})}$. Its description requires the information of all subsets $\hat{M}' \subseteq \hat{M}$:
\begin{equation}
    \varrho^{(\hat{M})} = \frac{1}{d^{(\hat{M})}} \qty( \mathbbm{1}_{\mathcal{H}^{(\hat{M})}} + \sum_{\hat{M}' \subseteq \hat{M}} \sum_{i^{(\hat{M}')} = 1}^{d^{(m)2} - 1} r_{i^{(\hat{M}')}} \lambda_{i^{(\hat{M}')}}) \, ,
\end{equation}
where $d^{(\hat{M})} \equiv \dim \mathcal{H}^{(\hat{M})} = \prod_{m \in \hat{M}'} d^{(m)}$, and $r_{0, \ldots, 0} = 1$. Let us define the vectors $\vec{r}^{(\hat{M})} := \qty(r_{i^{(\hat{M})}} | i^{(m)} \neq 0, \forall m \in \hat{M}) \in \mathbb{R}^{\prod_{m \in \hat{M}}d^{(m)^2 - 1}} \cong \bigotimes_{m \in \hat{M}} \mathbb{R}^{d^{(m)2} - 1}$, $\forall \hat{M} \subseteq \hat{N}$. Then, the description of $\varrho \in \mathcal{D(H)}$ is determined by those vectors: $\vec{r} = \bigoplus_{\hat{M} \subseteq \hat{N}} \vec{r}^{(\hat{M})} \in \mathbb{R}^{d^2 - 1}$. Alternatively, $\varrho^{(\hat{M})} \in \mathcal{D\qty[H^{(\hat{M})}]}$ is defined by the subset of vectors $\bigoplus_{\hat{M'} \subseteq \hat{M}} \vec{r}^{(\hat{M'})}$. Note that the direct sums has been taken in order of increasing cardinality. From now on, all results will be derived for reduced states, recovering global as a particular case setting $\hat{M} = \hat{N}$.

Performing the random measurements protocol, ignoring measurements on subsystems $\hat{N} \setminus \hat{M}$ (setting $A^{(n)} = \mathbbm{1}_{\mathcal{H}^{(n)}}$, $\forall n \notin \hat{M}$) one obtains the correlation distributions of subsystems $\hat{M}$. Repeating the process for every possible partition of the system and applying Eq.~\eqref{Equation: t-th order moments}, one retrieves the so-called marginal moments. Thus, the marginal 2nd-order moment of $\varrho$ over partition $\hat{M}$ is equal to the global 2nd-order moment of the reduced state $\varrho^{(\hat{M})}$, i.e.
$\mathcal{R} \qty[\varrho^{(\hat{M})}] \equiv \sum_{i^{(\hat{M})} = 1}^{d^{(m)2} - 1} r_{i^{(\hat{M})}}^2 = \norm{\vec{r}^{(\hat{M})}}^2$. Note that we have removed the superscript $(2)$ in $\mathcal{R}^{(2)}$, since 2nd-order moments will be used uniquely from now on. In addition, we remark that the 2nd-order marginal moments are related to the purity,
\begin{equation} \label{Equation: Relation marginal purities and moments}
    \mathcal{P}\qty[\varrho^{(\hat{M})}] \equiv \tr (\varrho^{(\hat{M})})^2 = \frac{1}{d^{(M)}} \qty( 1 + \sum_{\hat{M}' \subseteq \hat{M}} \norm{\vec{r}^{(\hat{M}')}}^2) \, .
\end{equation}
We refer to \ref{Proof: Relation marginal purities and moments} for details, while we note that $\mathcal{P}\qty[\varrho^{(\hat{M})}] \in \qty[\frac{1}{d^{(\hat{M})}}, 1]$, and the marginal 2nd-order moments take values in the range $\mathcal{R}\qty[\varrho^{(\hat{M})}] \in [0, d^{(\hat{M})}-1]$, $\forall \hat{M} \subseteq \hat{N}$.


Let the sets of all marginal 2nd-order moments and marginal purities of $\varrho$ with respect to all possible partitions $\hat{M}' \subseteq \hat{M}$ be respectively denoted by $\mathcal{R} \qty{\varrho^{(\hat{M})}} \equiv \qty{ \mathcal{R} \qty[\varrho^{(\hat{M}')}] \ \Big\vert \ \forall \hat{M}' \subseteq \hat{M}}$ and $\mathcal{P} \qty{\varrho^{(\hat{M})}} \equiv \qty{ \mathcal{P} \qty[\varrho^{(\hat{M}')}] \ \Big\vert \ \forall \hat{M}' \subseteq \hat{M}}$. The number of marginal moments (and purities) is $2^N$, although only $2^N - 1$ are non-trivial, since $r_{(0, \ldots, 0)} = 1$, $\forall \varrho \in \mathcal{D(H)}$. Indeed, $\mathcal{R} \qty{\varrho^{(\hat{M})}}$ and $\mathcal{P} \qty{\varrho^{(\hat{M})}}$ provide essentially the same information: $\mathcal{P} \qty{\varrho^{(\hat{M}')}} =\mathcal{P} \qty{\tilde{\varrho}^{(\hat{M}')}} \iff \mathcal{R} \qty{\varrho^{(\hat{M}')}} = \mathcal{R} \qty{\tilde{\varrho}^{(\hat{M}')}}$, where the equality holds among equal labeled elements (cf.~\ref{Proof: Relation set of purities and moments}).

{\em Marginal moments and Bloch rotations.--} Let $\mathcal{SO}(d^2 - 1)$ be the group of special orthogonal operators acting in $\mathbb{R}^{d^2 - 1}$. Then, the relation between states with equal purity (equal global moments) is as follows. If $\tilde{\varrho} = U \varrho U^\dagger$ with $U \in \mathcal{SU}(d)$, then $\exists Q \in \mathcal{SO}(d^2 - 1)$ such that $\vec{\tilde{r}} = Q \vec{r}$. The proof is in~\ref{Proof: Unitaries and rotations}, which can be generalized to arbitrary reduced states $\varrho^{(\hat{M})}$. This result suggests that there exist states $\varrho, \tilde{\varrho} \in \mathcal{D(H)}$ such that $\mathcal{P}(\varrho) = \mathcal{P}(\tilde{\varrho})$ or equivalently $\norm{\vec{\tilde{r}}} = \norm{\vec{r}}$ (i.e. $\vec{\tilde{r}} = Q \vec{r}$ with $Q \in \mathcal{SO}(d^2 - 1)$, and yet $\tilde{\varrho} \neq U \varrho U^\dagger$. The proof of this statement is in~\ref{proof: Same purity different moments}. As shown in~\ref{Proof: Same moment different classes}, having the same global moments does not imply being in the same local unitary class. Furthermore, we also observe that having the same purity, which is a function of the state which gathers information about the marginal moments of every partition, does not imply being in the same local unitary class.


In the Bloch vector representation, every state of dimension $d^{(\hat{M})}$ it is assigned to a unique vector $\mathbb{R}^{d^{(\hat{M})2} - 1}$, or $2^N - 1$ vectors $\vec{r}^{(\hat{M})}$ in the tensor product Hilbert space. However, the contrary is not true. The reason, although the Gell-Mann representation ensures $\tr \varrho = 1$ and $\varrho^\dagger = \varrho$, the positivity of the state is not ensured. For every completely positive trace preserving (CPTP) map that preserves the norm of the Bloch vector, there exists $Q \in \mathcal{SO}(d^2 - 1)$ such that $\vec{\tilde{r}} = Q \vec{r}$. The proof is in~\ref{Proof: Rotations and CPTP maps 2}.

{\em Marginal moments and entanglement.--} Bloch rotations are a useful tool to investigate the relation between LU-equivalent states and their set of marginal moments:
\begin{theorem} \label{Theorem: LU and LO}
    If $\tilde{\varrho}^{(\hat{M})} \sim \varrho^{(\hat{M})}$; then $\forall m \in \hat{M}$, $\exists Q^{(m)} \in \mathcal{SO}\qty[d^{(m)2} - 1]$ such that $\vec{\tilde{r}}^{(\hat{M}')} = \qty[\bigotimes_{m \in \hat{M}'} Q^{(m)}] \vec{r}^{(\hat{M}')}$, $\forall \hat{M}' \subseteq \hat{M}$.
\end{theorem}
\begin{proof}
    See End Matter (cf. \ref{Proof theorem: LU and LO}).
\end{proof}

This means that, given two LU-equivalent states, their Bloch vectors among any partition $\hat{M}$ are related by a tensor rotation, whose tensor terms act on the individual subsystems. Moreover, we also have
\begin{theorem} \label{Theorem: Equivalent states have the same moments}
     $\varrho^{(\hat{M})} \sim \tilde{\varrho}^{(\hat{M})} \implies \mathcal{R} \qty{\varrho^{(\hat{M})}} = \mathcal{R} \qty{\tilde{\varrho}^{(\hat{M})}}$.
\end{theorem}
\begin{proof}
    It is a direct consequence of the LU-invariance of $t$-th order moments. Two alternative proofs  are given in \ref{proof theorem: Equivalent states have the same moments}.
\end{proof}

%
%

Theorem~\ref{Theorem: Equivalent states have the same moments} indicates that $\mathcal{R} \qty{\varrho}$ may be useful in detecting entanglement, yet it is not enough to prove full distinguishability. For that, the other direction of implication in Th.~\ref{Theorem: Equivalent states have the same moments} is needed. We derive the conditions that the Bloch vectors must fulfill when two states display equal marginal moments:
\begin{theorem} \label{Theorem: Set of moments and rotations}
    If $\mathcal{R} \qty[\varrho^{(\hat{M})}] = \mathcal{R} \qty[\tilde{\varrho}^{(\hat{M})}]$, then $\exists Q^{(\hat{M}')} \in \mathcal{SO}\qty[\prod_{m \in \hat{M}'}d^{(m)^2 - 1}]$ such that $\vec{\tilde{r}}^{(\hat{M}')} = Q^{(\hat{M}')} \vec{r}^{(\hat{M}')}$.
\end{theorem}
\begin{proof}
    See End Matter (cf. \ref{Proof theorem: Set of moments and rotations}).
\end{proof}


Joining all our results up to now, we can construct the following chain of statements:
\begin{align}
    \mathcal{R} \{\varrho\} & = \mathcal{R} \{\tilde{\varrho}\}\label{S1}\\
    & \Updownarrow\nonumber\\
    \vec{\tilde{r}}^{(\hat{M})} = Q^{(\hat{M})} & \vec{r}^{(\hat{M})}, \ \forall \hat{M} \subseteq \hat{N}\label{S2}\\
    & \Uparrow\nonumber\\
    \vec{\tilde{r}}^{(\hat{M})} = \bigg[\bigotimes_{m \in \hat{M}} Q^{(m)}\bigg] \vec{r}^{(\hat{M})} &, \ \forall \hat{M} \subseteq \hat{N}, m \in \hat{M}\label{S3}\\
    & \Uparrow\nonumber\\
    \varrho &\sim \tilde{\varrho}\label{S4}
\end{align}
Note that~\eqref{S2}$\notimplies$\eqref{S3} since $\bigotimes_{m \in \hat{M}} \mathcal{SO}\qty[d^{(m)2} - 1] \subset \mathcal{SO} \qty[\prod_{m \in \hat{M}} d^{(m)2} - 1]$, while~\eqref{S3}$\notimplies$\eqref{S4} because of Th.~\ref{Theorem: LU and LO} which is a one way implication, since $\mathcal{SU}(d)$ and $\mathcal{SO}(d^2 - 1)$ are not isomorphic for a general $d$. Thus, for $N$ particles of dimension $d^{(n)}$, there exist states $\varrho, \tilde{\varrho} \in \mathcal{D(H)}$ such that $\mathcal{R} \qty{\varrho} = \mathcal{R} \qty{\tilde{\varrho}}$ and $\varrho \nsim \tilde{\varrho}$. As we will show in the following, this result holds even for qubits, which suffices to prove the general case. 

Therefore, from now on we focus on qubits, i.e. $d^{(n)} = 2$, $\forall n \in \hat{N}$. The two-to-one surjective homomorphism between $\mathcal{SU}(2)$ and $\mathcal{SO}(3)$ allows for~\eqref{S3}$\iff$\eqref{S4}:


%
%

\begin{theorem} \label{Theorem: Unitary equivalence reduced states}
    $\mathcal{R} \qty[\varrho^{(m)}] = \mathcal{R} \qty[\tilde{\varrho}^{(m)}] \iff \exists U^{(m)} \in \mathcal{SU}(2)$ such that $\tilde{\varrho}^{(m)} = U^{(m)} \varrho^{(m)} U^{(m)\dagger}$. 
\end{theorem}

\begin{proof}
    See \ref{Proof theorem: Unitary equivalence reduced states} for the proof and the complete characterization of these unitaries.
\end{proof}

Theorem \ref{Theorem: Unitary equivalence reduced states} allows to prove $\varrho \sim \tilde{\varrho} \iff \mathcal{R} (\varrho) = \mathcal{R} (\tilde{\varrho})$ for two special cases: $N=1$ with general mixed states, and $N = 2$ with pure states. For $N = 1$, unitary and LU equivalence are the same concept, so it follows trivially. For $N = 2$ pure states, the Schmidt decomposition allows to prove that $\mathcal{R}\qty[\varrho^{(1)}] = \mathcal{R}\qty[\varrho^{(2)}]$. Then from Eq.~\eqref{Equation: Relation marginal purities and moments} follows that only one of $\mathcal{R}\qty[\varrho^{(1)}]$, $\mathcal{R}\qty[\varrho^{(2)}]$ or $\mathcal{R}\qty[\varrho^{(1, 2)}]$ allows to completely determine $\mathcal{R}\qty{\varrho}$. Finally by Th.~\ref{Theorem: Unitary equivalence reduced states}, $\varrho$ and $\tilde{\varrho}$ will have the same Schmidt coefficients, and thus in agreement with~\cite{Tran_2015,Tran_2016}.

%
%

\begin{corollary} \label{Corollary: LU(2) and LO(3)}
    If $d^{(m)} = 2$, $\forall m \in \hat{M}'$, then Theorem~\ref{Theorem: LU and LO} is bidirectional, i.e. with an if and only if implication.
\end{corollary}

\begin{proof}
    See End Matter (cf. \ref{Proof theorem: LU and LO}).
\end{proof}


Even in the case of qubits, the relation~\eqref{S2}$\notimplies$\eqref{S3} still holds, so there exist states $\varrho, \tilde{\varrho} \in \mathcal{D(H)}$ such that $\mathcal{R} \qty{\varrho} = \mathcal{R} \qty{\tilde{\varrho}}$, and yet $\varrho \nsim \tilde{\varrho}$. To prove this, we develop a method for the generation of counterexamples: Evolve an initial state $\varrho$ with a rotation $Q = \bigoplus_{\hat{M} \subseteq \hat{N}} Q^{(\hat{M})}$, subjected to $Q \neq \bigotimes_{m \in \hat{M}} Q^{(m)}$, such that $\tilde{\varrho} \geq 0$. It is important to ensure that $2^{-N} < \mathcal{P}[\varrho] < 1$. If $\mathcal{P}[\varrho] = 2^{-N}$ the problem is trivial since $\norm{\vec{r}^{(\hat{M})}}^2 = 0$, $\forall \hat{M} \subseteq \hat{N}$; while for $\mathcal{P}[\varrho]=1$ one falls back to the pure state case.


Let us consider a particular counterexample using the following separable two-qubit state, $\varrho_{ce} = \frac{3}{4} \dyad{00} + \frac{1}{4} \dyad{11}$ so that $\vec{r}^{(1)}_{ce} = \vec{r}^{(2)}_{ce} = \hat{z} / 2$ and $\vec{r}^{(1,2)}_{ce} = \hat{z} \otimes \hat{z}$, with $\hat{x},\hat{y},\hat{z}$ denoting unit vectors of $\mathbb{R}^3$. Then, we perform the rotation $Q = Q^{(1)} \oplus Q^{(2)} \oplus Q^{(1, 2)}$ given by $Q^{(1)} = Q^{(2)} = \mathbbm{1}_{\mathbb{R}^3}$, while $Q^{(1,2)}$ transforms the Bloch vector as $\vec{\tilde{r}}^{(1,2)}_{ce} = Q^{(1,2)} \vec{r}^{(1,2)}_{ce} = (0, a, 0, b, 0, 0, 0, 0, c)^T$, with $a^2 + b^2 + c^2 = 1$. The transformed states $\tilde{\varrho}_{ce}$ are positive for $a - b - c \geq -1$, $b - a - c \geq -1$ and $2c(c + 1) \geq 2ab + 1$. To conclude the proof we need to ensure that $Q^{(1, 2)}$ cannot be expressed as $Q^{(1)} \otimes Q^{(2)}$.
\begin{observation} \label{Observation: Product rotation}
    No operator $Q^{(1, 2)}$, that performs the rotation given above, can be implemented as a product $Q^{(1, 2)} = Q^{(1)} \otimes Q^{(2)}$ leaving $\vec{r}^{(1)}_{ce}$ and $\vec{r}^{(2)}_{ce}$ invariant.
\end{observation}
\begin{proof}
    Establish that $Q^{(n)} \hat{z} = \hat{z}$, $\forall n \in \qty{1, 2}$, and $Q^{(1, 2)} \hat{z} \otimes \hat{z} = a \hat{x} \otimes \hat{y} + b \hat{y} \otimes \hat{x} + c \hat{z} \otimes \hat{z}$, with $a, b, c \neq 0$. Assume that $Q^{(1, 2)} = Q^{(1)} \otimes Q^{(2)}$, so that $Q^{(1, 2)} \hat{z} \otimes \hat{z} = Q^{(1)} \hat{z} \otimes Q^{(2)} \hat{z} = \hat{z} \otimes \hat{z}$. Since the decomposition of a vector into a basis is unique, $Q^{(1, 2)} = Q^{(1)} \otimes Q^{(2)}$ if and only if $a = b = 0$ and $c = 1$. This is a contradiction and thus the assumption is wrong: $Q^{(1, 2)} \neq Q^{(1)} \otimes Q^{(2)}$.
\end{proof}

Note that by construction, $\mathcal{R} \qty{\varrho_{ce}} = \mathcal{R} \qty{\tilde{\varrho}_{ce}}$. However, by Observation \ref{Observation: Product rotation} and Corollary \ref{Corollary: LU(2) and LO(3)} (or Theorem \ref{Theorem: LU and LO} in the general non-qubit case), $\varrho_{ce} \nsim \tilde{\varrho}_{ce}$ for $a,b\neq 0$. Different families of counterexamples arise from different choices of $\varrho_{ce}$, combinations of non-zero coefficients in $\vec{\tilde{r}}^{(1, 2)}_{ce}$, or rotations $Q$. Furthermore, while $\varrho_{ce}$ is separable, $\tilde{\varrho}_{ce}$ may contain different amounts of entanglement depending on the choice of coefficients $a$, $b$ and $c$. This can be seen by computing the well-known entanglement of formation, $E_F(\varrho)$~\cite{Hill_1997, Plenio_2007}. Clearly, for the separable state $E_F(\varrho_{ce}) = 0$. To the contrary, this family of states $\tilde{\varrho}_{ce}$ with $a,b\neq 0$  is entangled, with $0<E_F(\tilde{\varrho}_{ce})\leq E_F^{max}$ where the maximum takes place for $a=b=c=1/\sqrt{3}$ and amounts to $E_F^{max}\approx 0.22$. Therefore, we have proved that there exists states $\varrho, \tilde{\varrho} \in \mathcal{D}(\mathcal{H})$ such that $\mathcal{R} \qty{\varrho} = \mathcal{R} \qty{\tilde{\varrho}}$, yet they are not LU equivalent, $\varrho \nsim \tilde{\varrho}$. Notoriously, the states can be either separable or entangled, demonstrating that the full knowledge of marginal and global 2nd-order moments do not suffice for entanglement detection nor SLOCC classification. However, this does not apply to pure states, since the global 2nd-order moment achieves its lower value if and only if the state is separable~\cite{Tran_2015, Tran_2016}.

{\em Conclusions.--} In this Letter we show that full knowledge of the global and marginal second order moments of a state is insufficient to detect entanglement in a general setting. Indeed, two states may share identical second order moments and yet not be related by local-unitary (LU) transformations. In order to prove this statement we derive the conditions that two LU-equivalents states need to fulfill in terms of their Bloch vectors and how this manifests in their second-order moments. Building on these conditions, we construct a family of states that, even if they share identical global and marginal second-order moments among them, are not LU equivalent, being either separable or entangled with different amount. As a consequence, we also prove that the condition is, furthermore, not enough for the states to be SLOCC equivalent. Therefore, our work contains valuable results in the context of entanglement detection, underscoring the need of more robust criteria beyond second-order measurements, which is of paramount importance in quantum science and technologies.

\begin{acknowledgments}
We acknowledge financial support form the Spanish Government via the project PID2021-126694NA-C22 (MCIU/AEI/FEDER, EU) and project TSI-069100-2023-8 (Perte Chip-NextGenerationEU). E.T.,
R.P. and Y.B. acknowledge the Ram{\'o}n y Cajal (RYC2020-030060-I), (RYC2023-044095-I) and (RYC2023-042699-I) research fellowships. 
\end{acknowledgments}


\begin{thebibliography}{74}%
\makeatletter
\providecommand \@ifxundefined [1]{%
 \@ifx{#1\undefined}
}%
\providecommand \@ifnum [1]{%
 \ifnum #1\expandafter \@firstoftwo
 \else \expandafter \@secondoftwo
 \fi
}%
\providecommand \@ifx [1]{%
 \ifx #1\expandafter \@firstoftwo
 \else \expandafter \@secondoftwo
 \fi
}%
\providecommand \natexlab [1]{#1}%
\providecommand \enquote  [1]{``#1''}%
\providecommand \bibnamefont  [1]{#1}%
\providecommand \bibfnamefont [1]{#1}%
\providecommand \citenamefont [1]{#1}%
\providecommand \href@noop [0]{\@secondoftwo}%
\providecommand \href [0]{\begingroup \@sanitize@url \@href}%
\providecommand \@href[1]{\@@startlink{#1}\@@href}%
\providecommand \@@href[1]{\endgroup#1\@@endlink}%
\providecommand \@sanitize@url [0]{\catcode `\\12\catcode `\$12\catcode
  `\&12\catcode `\#12\catcode `\^12\catcode `\_12\catcode `\%12\relax}%
\providecommand \@@startlink[1]{}%
\providecommand \@@endlink[0]{}%
\providecommand \url  [0]{\begingroup\@sanitize@url \@url }%
\providecommand \@url [1]{\endgroup\@href {#1}{\urlprefix }}%
\providecommand \urlprefix  [0]{URL }%
\providecommand \Eprint [0]{\href }%
\providecommand \doibase [0]{https://doi.org/}%
\providecommand \selectlanguage [0]{\@gobble}%
\providecommand \bibinfo  [0]{\@secondoftwo}%
\providecommand \bibfield  [0]{\@secondoftwo}%
\providecommand \translation [1]{[#1]}%
\providecommand \BibitemOpen [0]{}%
\providecommand \bibitemStop [0]{}%
\providecommand \bibitemNoStop [0]{.\EOS\space}%
\providecommand \EOS [0]{\spacefactor3000\relax}%
\providecommand \BibitemShut  [1]{\csname bibitem#1\endcsname}%
\let\auto@bib@innerbib\@empty
\bibitem [{\citenamefont {Horodecki}\ \emph {et~al.}(2009)\citenamefont
  {Horodecki}, \citenamefont {Horodecki}, \citenamefont {Horodecki},\ and\
  \citenamefont {Horodecki}}]{Horodecki_2009}%
  \BibitemOpen
  \bibfield  {author} {\bibinfo {author} {\bibfnamefont {R.}~\bibnamefont
  {Horodecki}}, \bibinfo {author} {\bibfnamefont {P.}~\bibnamefont
  {Horodecki}}, \bibinfo {author} {\bibfnamefont {M.}~\bibnamefont
  {Horodecki}},\ and\ \bibinfo {author} {\bibfnamefont {K.}~\bibnamefont
  {Horodecki}},\ }\bibfield  {title} {\bibinfo {title} {Quantum entanglement},\
  }\href {https://doi.org/10.1103/RevModPhys.81.865} {\bibfield  {journal}
  {\bibinfo  {journal} {Rev. Mod. Phys.}\ }\textbf {\bibinfo {volume} {81}},\
  \bibinfo {pages} {865} (\bibinfo {year} {2009})}\BibitemShut {NoStop}%
\bibitem [{\citenamefont {Einstein}\ \emph {et~al.}(1935)\citenamefont
  {Einstein}, \citenamefont {Podolsky},\ and\ \citenamefont
  {Rosen}}]{Einstein_35}%
  \BibitemOpen
  \bibfield  {author} {\bibinfo {author} {\bibfnamefont {A.}~\bibnamefont
  {Einstein}}, \bibinfo {author} {\bibfnamefont {B.}~\bibnamefont {Podolsky}},\
  and\ \bibinfo {author} {\bibfnamefont {N.}~\bibnamefont {Rosen}},\ }\bibfield
   {title} {\bibinfo {title} {Can quantum-mechanical description of physical
  reality be considered complete?},\ }\href
  {https://doi.org/10.1103/PhysRev.47.777} {\bibfield  {journal} {\bibinfo
  {journal} {Phys. Rev.}\ }\textbf {\bibinfo {volume} {47}},\ \bibinfo {pages}
  {777} (\bibinfo {year} {1935})}\BibitemShut {NoStop}%
\bibitem [{\citenamefont {Bell}(1964)}]{Bell_64}%
  \BibitemOpen
  \bibfield  {author} {\bibinfo {author} {\bibfnamefont {J.~S.}\ \bibnamefont
  {Bell}},\ }\bibfield  {title} {\bibinfo {title} {On the
  {Einstein-Poldolsky-Rosen} paradox},\ }\href
  {https://doi.org/10.1103/PhysicsPhysiqueFizika.1.195} {\bibfield  {journal}
  {\bibinfo  {journal} {Physics}\ }\textbf {\bibinfo {volume} {1}},\ \bibinfo
  {pages} {195} (\bibinfo {year} {1964})}\BibitemShut {NoStop}%
\bibitem [{\citenamefont {Brunner}\ \emph {et~al.}(2014)\citenamefont
  {Brunner}, \citenamefont {Cavalcanti}, \citenamefont {Pironio}, \citenamefont
  {Scarani},\ and\ \citenamefont {Wehner}}]{Brunner_2014}%
  \BibitemOpen
  \bibfield  {author} {\bibinfo {author} {\bibfnamefont {N.}~\bibnamefont
  {Brunner}}, \bibinfo {author} {\bibfnamefont {D.}~\bibnamefont {Cavalcanti}},
  \bibinfo {author} {\bibfnamefont {S.}~\bibnamefont {Pironio}}, \bibinfo
  {author} {\bibfnamefont {V.}~\bibnamefont {Scarani}},\ and\ \bibinfo {author}
  {\bibfnamefont {S.}~\bibnamefont {Wehner}},\ }\bibfield  {title} {\bibinfo
  {title} {Bell nonlocality},\ }\href
  {https://doi.org/10.1103/RevModPhys.86.419} {\bibfield  {journal} {\bibinfo
  {journal} {Rev. Mod. Phys.}\ }\textbf {\bibinfo {volume} {86}},\ \bibinfo
  {pages} {419} (\bibinfo {year} {2014})}\BibitemShut {NoStop}%
\bibitem [{\citenamefont {Dowling}\ and\ \citenamefont
  {Milburn}(2003)}]{Dowling_2003}%
  \BibitemOpen
  \bibfield  {author} {\bibinfo {author} {\bibfnamefont {J.~P.}\ \bibnamefont
  {Dowling}}\ and\ \bibinfo {author} {\bibfnamefont {G.~J.}\ \bibnamefont
  {Milburn}},\ }\bibfield  {title} {\bibinfo {title} {{Quantum technology: the
  second quantum revolution}},\ }\href {https://doi.org/10.1098/rsta.2003.1227}
  {\bibfield  {journal} {\bibinfo  {journal} {Phil. Trans. R. Soc. A}\ }\textbf
  {\bibinfo {volume} {361}},\ \bibinfo {pages} {1655} (\bibinfo {year}
  {2003})}\BibitemShut {NoStop}%
\bibitem [{\citenamefont {Nielsen}\ and\ \citenamefont
  {Chuang}(2010)}]{Nielsen_2010}%
  \BibitemOpen
  \bibfield  {author} {\bibinfo {author} {\bibfnamefont {M.~A.}\ \bibnamefont
  {Nielsen}}\ and\ \bibinfo {author} {\bibfnamefont {I.~L.}\ \bibnamefont
  {Chuang}},\ }\href@noop {} {\emph {\bibinfo {title} {Quantum Computation and
  Quantum Information: 10th Anniversary Edition}}}\ (\bibinfo  {publisher}
  {Cambridge University Press},\ \bibinfo {year} {2010})\BibitemShut {NoStop}%
\bibitem [{\citenamefont {Paris}(2009)}]{Paris_09}%
  \BibitemOpen
  \bibfield  {author} {\bibinfo {author} {\bibfnamefont {M.~G.~A.}\
  \bibnamefont {Paris}},\ }\bibfield  {title} {\bibinfo {title} {Quantum
  estimation for quantum technology},\ }\href
  {https://doi.org/10.1142/S0219749909004839} {\bibfield  {journal} {\bibinfo
  {journal} {Int. J. Quant. Inf.}\ }\textbf {\bibinfo {volume} {7}},\ \bibinfo
  {pages} {125} (\bibinfo {year} {2009})}\BibitemShut {NoStop}%
\bibitem [{\citenamefont {Giovannetti}\ \emph {et~al.}(2011)\citenamefont
  {Giovannetti}, \citenamefont {Lloyd},\ and\ \citenamefont
  {Maccone}}]{Giovannetti_2011}%
  \BibitemOpen
  \bibfield  {author} {\bibinfo {author} {\bibfnamefont {V.}~\bibnamefont
  {Giovannetti}}, \bibinfo {author} {\bibfnamefont {S.}~\bibnamefont {Lloyd}},\
  and\ \bibinfo {author} {\bibfnamefont {L.}~\bibnamefont {Maccone}},\
  }\bibfield  {title} {\bibinfo {title} {Advances in quantum metrology},\
  }\href {https://doi.org/10.1038/nphoton.2011.35} {\bibfield  {journal}
  {\bibinfo  {journal} {Nat. Phot.}\ }\textbf {\bibinfo {volume} {5}},\
  \bibinfo {pages} {222} (\bibinfo {year} {2011})}\BibitemShut {NoStop}%
\bibitem [{\citenamefont {Friis}\ \emph {et~al.}(2019)\citenamefont {Friis},
  \citenamefont {Vitagliano}, \citenamefont {Malik},\ and\ \citenamefont
  {Huber}}]{Friis_2019}%
  \BibitemOpen
  \bibfield  {author} {\bibinfo {author} {\bibfnamefont {N.}~\bibnamefont
  {Friis}}, \bibinfo {author} {\bibfnamefont {G.}~\bibnamefont {Vitagliano}},
  \bibinfo {author} {\bibfnamefont {M.}~\bibnamefont {Malik}},\ and\ \bibinfo
  {author} {\bibfnamefont {M.}~\bibnamefont {Huber}},\ }\bibfield  {title}
  {\bibinfo {title} {Entanglement certification from theory to experiment},\
  }\href {https://doi.org/10.1038/s42254-018-0003-5} {\bibfield  {journal}
  {\bibinfo  {journal} {Nat. Rev. Phys.}\ }\textbf {\bibinfo {volume} {1}},\
  \bibinfo {pages} {72} (\bibinfo {year} {2019})}\BibitemShut {NoStop}%
\bibitem [{\citenamefont {Blume-Kohout}(2010)}]{Blume-Kohout_2010}%
  \BibitemOpen
  \bibfield  {author} {\bibinfo {author} {\bibfnamefont {R.}~\bibnamefont
  {Blume-Kohout}},\ }\bibfield  {title} {\bibinfo {title} {Optimal, reliable
  estimation of quantum states},\ }\href
  {https://doi.org/10.1088/1367-2630/12/4/043034} {\bibfield  {journal}
  {\bibinfo  {journal} {New J. Phys.}\ }\textbf {\bibinfo {volume} {12}},\
  \bibinfo {pages} {043034} (\bibinfo {year} {2010})}\BibitemShut {NoStop}%
\bibitem [{\citenamefont {Gross}\ \emph {et~al.}(2010)\citenamefont {Gross},
  \citenamefont {Liu}, \citenamefont {Flammia}, \citenamefont {Becker},\ and\
  \citenamefont {Eisert}}]{Gross_2010}%
  \BibitemOpen
  \bibfield  {author} {\bibinfo {author} {\bibfnamefont {D.}~\bibnamefont
  {Gross}}, \bibinfo {author} {\bibfnamefont {Y.-K.}\ \bibnamefont {Liu}},
  \bibinfo {author} {\bibfnamefont {S.~T.}\ \bibnamefont {Flammia}}, \bibinfo
  {author} {\bibfnamefont {S.}~\bibnamefont {Becker}},\ and\ \bibinfo {author}
  {\bibfnamefont {J.}~\bibnamefont {Eisert}},\ }\bibfield  {title} {\bibinfo
  {title} {Quantum state tomography via compressed sensing},\ }\href
  {https://doi.org/10.1103/PhysRevLett.105.150401} {\bibfield  {journal}
  {\bibinfo  {journal} {Phys. Rev. Lett.}\ }\textbf {\bibinfo {volume} {105}},\
  \bibinfo {pages} {150401} (\bibinfo {year} {2010})}\BibitemShut {NoStop}%
\bibitem [{\citenamefont {Somma}\ \emph {et~al.}(2006)\citenamefont {Somma},
  \citenamefont {Chiaverini},\ and\ \citenamefont {Berkeland}}]{Somma_2006}%
  \BibitemOpen
  \bibfield  {author} {\bibinfo {author} {\bibfnamefont {R.~D.}\ \bibnamefont
  {Somma}}, \bibinfo {author} {\bibfnamefont {J.}~\bibnamefont {Chiaverini}},\
  and\ \bibinfo {author} {\bibfnamefont {D.~J.}\ \bibnamefont {Berkeland}},\
  }\bibfield  {title} {\bibinfo {title} {Lower bounds for the fidelity of
  entangled-state preparation},\ }\href
  {https://doi.org/10.1103/PhysRevA.74.052302} {\bibfield  {journal} {\bibinfo
  {journal} {Phys. Rev. A}\ }\textbf {\bibinfo {volume} {74}},\ \bibinfo
  {pages} {052302} (\bibinfo {year} {2006})}\BibitemShut {NoStop}%
\bibitem [{\citenamefont {G\"uhne}\ \emph {et~al.}(2007)\citenamefont
  {G\"uhne}, \citenamefont {Lu}, \citenamefont {Gao},\ and\ \citenamefont
  {Pan}}]{Guhne_2007}%
  \BibitemOpen
  \bibfield  {author} {\bibinfo {author} {\bibfnamefont {O.}~\bibnamefont
  {G\"uhne}}, \bibinfo {author} {\bibfnamefont {C.-Y.}\ \bibnamefont {Lu}},
  \bibinfo {author} {\bibfnamefont {W.-B.}\ \bibnamefont {Gao}},\ and\ \bibinfo
  {author} {\bibfnamefont {J.-W.}\ \bibnamefont {Pan}},\ }\bibfield  {title}
  {\bibinfo {title} {Toolbox for entanglement detection and fidelity
  estimation},\ }\href {https://doi.org/10.1103/PhysRevA.76.030305} {\bibfield
  {journal} {\bibinfo  {journal} {Phys. Rev. A}\ }\textbf {\bibinfo {volume}
  {76}},\ \bibinfo {pages} {030305} (\bibinfo {year} {2007})}\BibitemShut
  {NoStop}%
\bibitem [{\citenamefont {Gühne}\ and\ \citenamefont
  {Tóth}(2009)}]{Guhne_2009}%
  \BibitemOpen
  \bibfield  {author} {\bibinfo {author} {\bibfnamefont {O.}~\bibnamefont
  {Gühne}}\ and\ \bibinfo {author} {\bibfnamefont {G.}~\bibnamefont {Tóth}},\
  }\bibfield  {title} {\bibinfo {title} {Entanglement detection},\ }\href
  {https://doi.org/https://doi.org/10.1016/j.physrep.2009.02.004} {\bibfield
  {journal} {\bibinfo  {journal} {Phys. Rep.}\ }\textbf {\bibinfo {volume}
  {474}},\ \bibinfo {pages} {1} (\bibinfo {year} {2009})}\BibitemShut {NoStop}%
\bibitem [{\citenamefont {Cramer}\ \emph {et~al.}(2010)\citenamefont {Cramer},
  \citenamefont {Plenio}, \citenamefont {Flammia}, \citenamefont {Somma},
  \citenamefont {Gross}, \citenamefont {Bartlett}, \citenamefont
  {Landon-Cardinal}, \citenamefont {Poulin},\ and\ \citenamefont
  {Liu}}]{Cramer_2010}%
  \BibitemOpen
  \bibfield  {author} {\bibinfo {author} {\bibfnamefont {M.}~\bibnamefont
  {Cramer}}, \bibinfo {author} {\bibfnamefont {M.~B.}\ \bibnamefont {Plenio}},
  \bibinfo {author} {\bibfnamefont {S.~T.}\ \bibnamefont {Flammia}}, \bibinfo
  {author} {\bibfnamefont {R.}~\bibnamefont {Somma}}, \bibinfo {author}
  {\bibfnamefont {D.}~\bibnamefont {Gross}}, \bibinfo {author} {\bibfnamefont
  {S.~D.}\ \bibnamefont {Bartlett}}, \bibinfo {author} {\bibfnamefont
  {O.}~\bibnamefont {Landon-Cardinal}}, \bibinfo {author} {\bibfnamefont
  {D.}~\bibnamefont {Poulin}},\ and\ \bibinfo {author} {\bibfnamefont {Y.-K.}\
  \bibnamefont {Liu}},\ }\bibfield  {title} {\bibinfo {title} {Efficient
  quantum state tomography},\ }\href {https://doi.org/10.1038/ncomms1147}
  {\bibfield  {journal} {\bibinfo  {journal} {Nat. Commun.}\ }\textbf {\bibinfo
  {volume} {1}},\ \bibinfo {pages} {149} (\bibinfo {year} {2010})}\BibitemShut
  {NoStop}%
\bibitem [{\citenamefont {{\v{S}}upi{\'{c}}}\ and\ \citenamefont
  {Bowles}(2020)}]{Bowles_2020}%
  \BibitemOpen
  \bibfield  {author} {\bibinfo {author} {\bibfnamefont {I.}~\bibnamefont
  {{\v{S}}upi{\'{c}}}}\ and\ \bibinfo {author} {\bibfnamefont {J.}~\bibnamefont
  {Bowles}},\ }\bibfield  {title} {\bibinfo {title} {Self-testing of quantum
  systems: a review},\ }\href {https://doi.org/10.22331/q-2020-09-30-337}
  {\bibfield  {journal} {\bibinfo  {journal} {{Quantum}}\ }\textbf {\bibinfo
  {volume} {4}},\ \bibinfo {pages} {337} (\bibinfo {year} {2020})}\BibitemShut
  {NoStop}%
\bibitem [{\citenamefont {Bartlett}\ \emph {et~al.}(2007)\citenamefont
  {Bartlett}, \citenamefont {Rudolph},\ and\ \citenamefont
  {Spekkens}}]{Bartlett_2007}%
  \BibitemOpen
  \bibfield  {author} {\bibinfo {author} {\bibfnamefont {S.~D.}\ \bibnamefont
  {Bartlett}}, \bibinfo {author} {\bibfnamefont {T.}~\bibnamefont {Rudolph}},\
  and\ \bibinfo {author} {\bibfnamefont {R.~W.}\ \bibnamefont {Spekkens}},\
  }\bibfield  {title} {\bibinfo {title} {Reference frames, superselection
  rules, and quantum information},\ }\href
  {https://doi.org/10.1103/RevModPhys.79.555} {\bibfield  {journal} {\bibinfo
  {journal} {Rev. Mod. Phys.}\ }\textbf {\bibinfo {volume} {79}},\ \bibinfo
  {pages} {555} (\bibinfo {year} {2007})}\BibitemShut {NoStop}%
\bibitem [{\citenamefont {Walter}\ \emph {et~al.}(2013)\citenamefont {Walter},
  \citenamefont {Doran}, \citenamefont {Gross},\ and\ \citenamefont
  {Christandl}}]{Walter_2013}%
  \BibitemOpen
  \bibfield  {author} {\bibinfo {author} {\bibfnamefont {M.}~\bibnamefont
  {Walter}}, \bibinfo {author} {\bibfnamefont {B.}~\bibnamefont {Doran}},
  \bibinfo {author} {\bibfnamefont {D.}~\bibnamefont {Gross}},\ and\ \bibinfo
  {author} {\bibfnamefont {M.}~\bibnamefont {Christandl}},\ }\bibfield  {title}
  {\bibinfo {title} {Entanglement polytopes: Multiparticle entanglement from
  single-particle information},\ }\href
  {https://doi.org/10.1126/science.1232957} {\bibfield  {journal} {\bibinfo
  {journal} {Science}\ }\textbf {\bibinfo {volume} {340}},\ \bibinfo {pages}
  {1205} (\bibinfo {year} {2013})}\BibitemShut {NoStop}%
\bibitem [{\citenamefont {Aguilar}\ \emph {et~al.}(2015)\citenamefont
  {Aguilar}, \citenamefont {Walborn}, \citenamefont {Ribeiro},\ and\
  \citenamefont {C\'eleri}}]{Aguilar_2015}%
  \BibitemOpen
  \bibfield  {author} {\bibinfo {author} {\bibfnamefont {G.~H.}\ \bibnamefont
  {Aguilar}}, \bibinfo {author} {\bibfnamefont {S.~P.}\ \bibnamefont
  {Walborn}}, \bibinfo {author} {\bibfnamefont {P.~H.~S.}\ \bibnamefont
  {Ribeiro}},\ and\ \bibinfo {author} {\bibfnamefont {L.~C.}\ \bibnamefont
  {C\'eleri}},\ }\bibfield  {title} {\bibinfo {title} {Experimental
  determination of multipartite entanglement with incomplete information},\
  }\href {https://doi.org/10.1103/PhysRevX.5.031042} {\bibfield  {journal}
  {\bibinfo  {journal} {Phys. Rev. X}\ }\textbf {\bibinfo {volume} {5}},\
  \bibinfo {pages} {031042} (\bibinfo {year} {2015})}\BibitemShut {NoStop}%
\bibitem [{\citenamefont {Aolita}\ and\ \citenamefont
  {Walborn}(2007)}]{Aolita_2007}%
  \BibitemOpen
  \bibfield  {author} {\bibinfo {author} {\bibfnamefont {L.}~\bibnamefont
  {Aolita}}\ and\ \bibinfo {author} {\bibfnamefont {S.~P.}\ \bibnamefont
  {Walborn}},\ }\bibfield  {title} {\bibinfo {title} {Quantum communication
  without alignment using multiple-qubit single-photon states},\ }\href
  {https://doi.org/10.1103/PhysRevLett.98.100501} {\bibfield  {journal}
  {\bibinfo  {journal} {Phys. Rev. Lett.}\ }\textbf {\bibinfo {volume} {98}},\
  \bibinfo {pages} {100501} (\bibinfo {year} {2007})}\BibitemShut {NoStop}%
\bibitem [{\citenamefont {D'Ambrosio}\ \emph {et~al.}(2012)\citenamefont
  {D'Ambrosio}, \citenamefont {Nagali}, \citenamefont {Walborn}, \citenamefont
  {Aolita}, \citenamefont {Slussarenko}, \citenamefont {Marrucci},\ and\
  \citenamefont {Sciarrino}}]{D'Ambrosio_2012}%
  \BibitemOpen
  \bibfield  {author} {\bibinfo {author} {\bibfnamefont {V.}~\bibnamefont
  {D'Ambrosio}}, \bibinfo {author} {\bibfnamefont {E.}~\bibnamefont {Nagali}},
  \bibinfo {author} {\bibfnamefont {S.~P.}\ \bibnamefont {Walborn}}, \bibinfo
  {author} {\bibfnamefont {L.}~\bibnamefont {Aolita}}, \bibinfo {author}
  {\bibfnamefont {S.}~\bibnamefont {Slussarenko}}, \bibinfo {author}
  {\bibfnamefont {L.}~\bibnamefont {Marrucci}},\ and\ \bibinfo {author}
  {\bibfnamefont {F.}~\bibnamefont {Sciarrino}},\ }\bibfield  {title} {\bibinfo
  {title} {Complete experimental toolbox for alignment-free quantum
  communication},\ }\href {https://doi.org/10.1038/ncomms1951} {\bibfield
  {journal} {\bibinfo  {journal} {Nat. Commun.}\ }\textbf {\bibinfo {volume}
  {3}},\ \bibinfo {pages} {961} (\bibinfo {year} {2012})}\BibitemShut {NoStop}%
\bibitem [{\citenamefont {Aschauer}\ \emph {et~al.}(2004)\citenamefont
  {Aschauer}, \citenamefont {Calsamiglia}, \citenamefont {Hein},\ and\
  \citenamefont {Briegel}}]{Aschauer_2004}%
  \BibitemOpen
  \bibfield  {author} {\bibinfo {author} {\bibfnamefont {H.}~\bibnamefont
  {Aschauer}}, \bibinfo {author} {\bibfnamefont {J.}~\bibnamefont
  {Calsamiglia}}, \bibinfo {author} {\bibfnamefont {M.}~\bibnamefont {Hein}},\
  and\ \bibinfo {author} {\bibfnamefont {H.~J.}\ \bibnamefont {Briegel}},\
  }\bibfield  {title} {\bibinfo {title} {Local invariants for multi-partite
  entangled states allowing for a simple entanglement criterion},\ }\href
  {https://dl.acm.org/doi/abs/10.5555/2011586.2011590} {\bibfield  {journal}
  {\bibinfo  {journal} {Quantum Info. Comput.}\ }\textbf {\bibinfo {volume}
  {4}},\ \bibinfo {pages} {383} (\bibinfo {year} {2004})}\BibitemShut {NoStop}%
\bibitem [{\citenamefont {De~Vicente}(2007)}]{Vicente_2007}%
  \BibitemOpen
  \bibfield  {author} {\bibinfo {author} {\bibfnamefont {J.~I.}\ \bibnamefont
  {De~Vicente}},\ }\bibfield  {title} {\bibinfo {title} {Separability criteria
  based on the {B}loch representation of density matrices},\ }\href
  {https://dl.acm.org/doi/abs/10.5555/2011734.2011739} {\bibfield  {journal}
  {\bibinfo  {journal} {Quantum Info. Comput.}\ }\textbf {\bibinfo {volume}
  {7}},\ \bibinfo {pages} {624} (\bibinfo {year} {2007})}\BibitemShut {NoStop}%
\bibitem [{\citenamefont {de~Vicente}(2008)}]{Vicente_2008}%
  \BibitemOpen
  \bibfield  {author} {\bibinfo {author} {\bibfnamefont {J.~I.}\ \bibnamefont
  {de~Vicente}},\ }\bibfield  {title} {\bibinfo {title} {Further results on
  entanglement detection and quantification from the correlation matrix
  criterion},\ }\href {https://doi.org/10.1088/1751-8113/41/6/065309}
  {\bibfield  {journal} {\bibinfo  {journal} {J. Phys. A: Math. Theor.}\
  }\textbf {\bibinfo {volume} {41}},\ \bibinfo {pages} {065309} (\bibinfo
  {year} {2008})}\BibitemShut {NoStop}%
\bibitem [{\citenamefont {Badzia\ifmmode~\mbox{\c{}}\else \c{}\fi{}g}\ \emph
  {et~al.}(2008)\citenamefont {Badzia\ifmmode~\mbox{\c{}}\else \c{}\fi{}g},
  \citenamefont {Brukner}, \citenamefont {Laskowski}, \citenamefont {Paterek},\
  and\ \citenamefont {\ifmmode~\dot{Z}\else \.{Z}\fi{}ukowski}}]{Badziag_2008}%
  \BibitemOpen
  \bibfield  {author} {\bibinfo {author} {\bibfnamefont {P.}~\bibnamefont
  {Badzia\ifmmode~\mbox{\c{}}\else \c{}\fi{}g}}, \bibinfo {author}
  {\bibfnamefont {{\v{C}}.}~\bibnamefont {Brukner}}, \bibinfo {author}
  {\bibfnamefont {W.}~\bibnamefont {Laskowski}}, \bibinfo {author}
  {\bibfnamefont {T.}~\bibnamefont {Paterek}},\ and\ \bibinfo {author}
  {\bibfnamefont {M.}~\bibnamefont {\ifmmode~\dot{Z}\else \.{Z}\fi{}ukowski}},\
  }\bibfield  {title} {\bibinfo {title} {Experimentally friendly geometrical
  criteria for entanglement},\ }\href
  {https://doi.org/10.1103/PhysRevLett.100.140403} {\bibfield  {journal}
  {\bibinfo  {journal} {Phys. Rev. Lett.}\ }\textbf {\bibinfo {volume} {100}},\
  \bibinfo {pages} {140403} (\bibinfo {year} {2008})}\BibitemShut {NoStop}%
\bibitem [{\citenamefont {Hassan}\ and\ \citenamefont
  {Joag}(2008)}]{Hassan_2008}%
  \BibitemOpen
  \bibfield  {author} {\bibinfo {author} {\bibfnamefont {A.~S.~M.}\
  \bibnamefont {Hassan}}\ and\ \bibinfo {author} {\bibfnamefont {P.~S.}\
  \bibnamefont {Joag}},\ }\bibfield  {title} {\bibinfo {title} {Separability
  criterion for multipartite quantum states based on the bloch representation
  of density matrices},\ }\href
  {https://dl.acm.org/doi/10.5555/2017011.2017018} {\bibfield  {journal}
  {\bibinfo  {journal} {Quantum Info. Comput.}\ }\textbf {\bibinfo {volume}
  {8}},\ \bibinfo {pages} {773} (\bibinfo {year} {2008})}\BibitemShut {NoStop}%
\bibitem [{\citenamefont {de~Vicente}\ and\ \citenamefont
  {Huber}(2011)}]{Vicente_2011}%
  \BibitemOpen
  \bibfield  {author} {\bibinfo {author} {\bibfnamefont {J.~I.}\ \bibnamefont
  {de~Vicente}}\ and\ \bibinfo {author} {\bibfnamefont {M.}~\bibnamefont
  {Huber}},\ }\bibfield  {title} {\bibinfo {title} {Multipartite entanglement
  detection from correlation tensors},\ }\href
  {https://doi.org/10.1103/PhysRevA.84.062306} {\bibfield  {journal} {\bibinfo
  {journal} {Phys. Rev. A}\ }\textbf {\bibinfo {volume} {84}},\ \bibinfo
  {pages} {062306} (\bibinfo {year} {2011})}\BibitemShut {NoStop}%
\bibitem [{\citenamefont {Laskowski}\ \emph {et~al.}(2011)\citenamefont
  {Laskowski}, \citenamefont {Markiewicz}, \citenamefont {Paterek},\ and\
  \citenamefont {\ifmmode~\dot{Z}\else \.{Z}\fi{}ukowski}}]{Laskowski_2011}%
  \BibitemOpen
  \bibfield  {author} {\bibinfo {author} {\bibfnamefont {W.}~\bibnamefont
  {Laskowski}}, \bibinfo {author} {\bibfnamefont {M.}~\bibnamefont
  {Markiewicz}}, \bibinfo {author} {\bibfnamefont {T.}~\bibnamefont
  {Paterek}},\ and\ \bibinfo {author} {\bibfnamefont {M.}~\bibnamefont
  {\ifmmode~\dot{Z}\else \.{Z}\fi{}ukowski}},\ }\bibfield  {title} {\bibinfo
  {title} {Correlation-tensor criteria for genuine multiqubit entanglement},\
  }\href {https://doi.org/10.1103/PhysRevA.84.062305} {\bibfield  {journal}
  {\bibinfo  {journal} {Phys. Rev. A}\ }\textbf {\bibinfo {volume} {84}},\
  \bibinfo {pages} {062305} (\bibinfo {year} {2011})}\BibitemShut {NoStop}%
\bibitem [{\citenamefont {Gabriel}\ \emph {et~al.}(2013)\citenamefont
  {Gabriel}, \citenamefont {Rudnicki},\ and\ \citenamefont
  {Hiesmayr}}]{Gabriel_2013}%
  \BibitemOpen
  \bibfield  {author} {\bibinfo {author} {\bibfnamefont {A.}~\bibnamefont
  {Gabriel}}, \bibinfo {author} {\bibfnamefont {L.}~\bibnamefont {Rudnicki}},\
  and\ \bibinfo {author} {\bibfnamefont {B.~C.}\ \bibnamefont {Hiesmayr}},\
  }\bibfield  {title} {\bibinfo {title} {Measurement-base-independent test for
  genuine multipartite entanglement},\ }\href
  {https://doi.org/10.1088/1367-2630/15/7/073033} {\bibfield  {journal}
  {\bibinfo  {journal} {New J. Phys.}\ }\textbf {\bibinfo {volume} {15}},\
  \bibinfo {pages} {073033} (\bibinfo {year} {2013})}\BibitemShut {NoStop}%
\bibitem [{\citenamefont {Liang}\ \emph {et~al.}(2010)\citenamefont {Liang},
  \citenamefont {Harrigan}, \citenamefont {Bartlett},\ and\ \citenamefont
  {Rudolph}}]{Liang_2010}%
  \BibitemOpen
  \bibfield  {author} {\bibinfo {author} {\bibfnamefont {Y.-C.}\ \bibnamefont
  {Liang}}, \bibinfo {author} {\bibfnamefont {N.}~\bibnamefont {Harrigan}},
  \bibinfo {author} {\bibfnamefont {S.~D.}\ \bibnamefont {Bartlett}},\ and\
  \bibinfo {author} {\bibfnamefont {T.}~\bibnamefont {Rudolph}},\ }\bibfield
  {title} {\bibinfo {title} {Nonclassical correlations from randomly chosen
  local measurements},\ }\href {https://doi.org/10.1103/PhysRevLett.104.050401}
  {\bibfield  {journal} {\bibinfo  {journal} {Phys. Rev. Lett.}\ }\textbf
  {\bibinfo {volume} {104}},\ \bibinfo {pages} {050401} (\bibinfo {year}
  {2010})}\BibitemShut {NoStop}%
\bibitem [{\citenamefont {Wallman}\ \emph {et~al.}(2011)\citenamefont
  {Wallman}, \citenamefont {Liang},\ and\ \citenamefont
  {Bartlett}}]{Wallman_2011}%
  \BibitemOpen
  \bibfield  {author} {\bibinfo {author} {\bibfnamefont {J.~J.}\ \bibnamefont
  {Wallman}}, \bibinfo {author} {\bibfnamefont {Y.-C.}\ \bibnamefont {Liang}},\
  and\ \bibinfo {author} {\bibfnamefont {S.~D.}\ \bibnamefont {Bartlett}},\
  }\bibfield  {title} {\bibinfo {title} {Generating nonclassical correlations
  without fully aligning measurements},\ }\href
  {https://doi.org/10.1103/PhysRevA.83.022110} {\bibfield  {journal} {\bibinfo
  {journal} {Phys. Rev. A}\ }\textbf {\bibinfo {volume} {83}},\ \bibinfo
  {pages} {022110} (\bibinfo {year} {2011})}\BibitemShut {NoStop}%
\bibitem [{\citenamefont {Shadbolt}\ \emph
  {et~al.}(2012{\natexlab{a}})\citenamefont {Shadbolt}, \citenamefont {Verde},
  \citenamefont {Peruzzo}, \citenamefont {Politi}, \citenamefont {Laing},
  \citenamefont {Lobino}, \citenamefont {Matthews}, \citenamefont {Thompson},\
  and\ \citenamefont {O'Brien}}]{Shadbolt_2012}%
  \BibitemOpen
  \bibfield  {author} {\bibinfo {author} {\bibfnamefont {P.~J.}\ \bibnamefont
  {Shadbolt}}, \bibinfo {author} {\bibfnamefont {M.~R.}\ \bibnamefont {Verde}},
  \bibinfo {author} {\bibfnamefont {A.}~\bibnamefont {Peruzzo}}, \bibinfo
  {author} {\bibfnamefont {A.}~\bibnamefont {Politi}}, \bibinfo {author}
  {\bibfnamefont {A.}~\bibnamefont {Laing}}, \bibinfo {author} {\bibfnamefont
  {M.}~\bibnamefont {Lobino}}, \bibinfo {author} {\bibfnamefont {J.~C.~F.}\
  \bibnamefont {Matthews}}, \bibinfo {author} {\bibfnamefont {M.~G.}\
  \bibnamefont {Thompson}},\ and\ \bibinfo {author} {\bibfnamefont {J.~L.}\
  \bibnamefont {O'Brien}},\ }\bibfield  {title} {\bibinfo {title} {Generating,
  manipulating and measuring entanglement and mixture with a reconfigurable
  photonic circuit},\ }\href {https://doi.org/10.1038/nphoton.2011.283}
  {\bibfield  {journal} {\bibinfo  {journal} {Nat. Phot.}\ }\textbf {\bibinfo
  {volume} {6}},\ \bibinfo {pages} {45} (\bibinfo {year}
  {2012}{\natexlab{a}})}\BibitemShut {NoStop}%
\bibitem [{\citenamefont {Wallman}\ and\ \citenamefont
  {Bartlett}(2012)}]{Wallman_2012}%
  \BibitemOpen
  \bibfield  {author} {\bibinfo {author} {\bibfnamefont {J.~J.}\ \bibnamefont
  {Wallman}}\ and\ \bibinfo {author} {\bibfnamefont {S.~D.}\ \bibnamefont
  {Bartlett}},\ }\bibfield  {title} {\bibinfo {title} {Observers can always
  generate nonlocal correlations without aligning measurements by covering all
  their bases},\ }\href {https://doi.org/10.1103/PhysRevA.85.024101} {\bibfield
   {journal} {\bibinfo  {journal} {Phys. Rev. A}\ }\textbf {\bibinfo {volume}
  {85}},\ \bibinfo {pages} {024101} (\bibinfo {year} {2012})}\BibitemShut
  {NoStop}%
\bibitem [{\citenamefont {Shadbolt}\ \emph
  {et~al.}(2012{\natexlab{b}})\citenamefont {Shadbolt}, \citenamefont
  {Vértesi}, \citenamefont {Liang}, \citenamefont {Branciard}, \citenamefont
  {Brunner},\ and\ \citenamefont {O'Brien}}]{Shadbolt_2012_2}%
  \BibitemOpen
  \bibfield  {author} {\bibinfo {author} {\bibfnamefont {P.}~\bibnamefont
  {Shadbolt}}, \bibinfo {author} {\bibfnamefont {T.}~\bibnamefont {Vértesi}},
  \bibinfo {author} {\bibfnamefont {Y.-C.}\ \bibnamefont {Liang}}, \bibinfo
  {author} {\bibfnamefont {C.}~\bibnamefont {Branciard}}, \bibinfo {author}
  {\bibfnamefont {N.}~\bibnamefont {Brunner}},\ and\ \bibinfo {author}
  {\bibfnamefont {J.~L.}\ \bibnamefont {O'Brien}},\ }\bibfield  {title}
  {\bibinfo {title} {Guaranteed violation of a {B}ell inequality without
  aligned reference frames or calibrated devices},\ }\href
  {https://doi.org/10.1038/srep00470} {\bibfield  {journal} {\bibinfo
  {journal} {Sci. Rep.}\ }\textbf {\bibinfo {volume} {2}},\ \bibinfo {pages}
  {470} (\bibinfo {year} {2012}{\natexlab{b}})}\BibitemShut {NoStop}%
\bibitem [{\citenamefont {Palsson}\ \emph {et~al.}(2012)\citenamefont
  {Palsson}, \citenamefont {Wallman}, \citenamefont {Bennet},\ and\
  \citenamefont {Pryde}}]{Palsson_2012}%
  \BibitemOpen
  \bibfield  {author} {\bibinfo {author} {\bibfnamefont {M.~S.}\ \bibnamefont
  {Palsson}}, \bibinfo {author} {\bibfnamefont {J.~J.}\ \bibnamefont
  {Wallman}}, \bibinfo {author} {\bibfnamefont {A.~J.}\ \bibnamefont
  {Bennet}},\ and\ \bibinfo {author} {\bibfnamefont {G.~J.}\ \bibnamefont
  {Pryde}},\ }\bibfield  {title} {\bibinfo {title} {Experimentally
  demonstrating reference-frame-independent violations of {B}ell
  inequalities},\ }\href {https://doi.org/10.1103/PhysRevA.86.032322}
  {\bibfield  {journal} {\bibinfo  {journal} {Phys. Rev. A}\ }\textbf {\bibinfo
  {volume} {86}},\ \bibinfo {pages} {032322} (\bibinfo {year}
  {2012})}\BibitemShut {NoStop}%
\bibitem [{\citenamefont {Furkan~Senel}\ \emph {et~al.}(2015)\citenamefont
  {Furkan~Senel}, \citenamefont {Lawson}, \citenamefont {Kaplan}, \citenamefont
  {Markham},\ and\ \citenamefont {Diamanti}}]{Senel_2015}%
  \BibitemOpen
  \bibfield  {author} {\bibinfo {author} {\bibfnamefont {C.}~\bibnamefont
  {Furkan~Senel}}, \bibinfo {author} {\bibfnamefont {T.}~\bibnamefont
  {Lawson}}, \bibinfo {author} {\bibfnamefont {M.}~\bibnamefont {Kaplan}},
  \bibinfo {author} {\bibfnamefont {D.}~\bibnamefont {Markham}},\ and\ \bibinfo
  {author} {\bibfnamefont {E.}~\bibnamefont {Diamanti}},\ }\bibfield  {title}
  {\bibinfo {title} {Demonstrating genuine multipartite entanglement and
  nonseparability without shared reference frames},\ }\href
  {https://doi.org/10.1103/PhysRevA.91.052118} {\bibfield  {journal} {\bibinfo
  {journal} {Phys. Rev. A}\ }\textbf {\bibinfo {volume} {91}},\ \bibinfo
  {pages} {052118} (\bibinfo {year} {2015})}\BibitemShut {NoStop}%
\bibitem [{\citenamefont {de~Rosier}\ \emph {et~al.}(2017)\citenamefont
  {de~Rosier}, \citenamefont {Gruca}, \citenamefont {Parisio}, \citenamefont
  {V\'ertesi},\ and\ \citenamefont {Laskowski}}]{Rosier_2017}%
  \BibitemOpen
  \bibfield  {author} {\bibinfo {author} {\bibfnamefont {A.}~\bibnamefont
  {de~Rosier}}, \bibinfo {author} {\bibfnamefont {J.}~\bibnamefont {Gruca}},
  \bibinfo {author} {\bibfnamefont {F.}~\bibnamefont {Parisio}}, \bibinfo
  {author} {\bibfnamefont {T.}~\bibnamefont {V\'ertesi}},\ and\ \bibinfo
  {author} {\bibfnamefont {W.}~\bibnamefont {Laskowski}},\ }\bibfield  {title}
  {\bibinfo {title} {Multipartite nonlocality and random measurements},\ }\href
  {https://doi.org/10.1103/PhysRevA.96.012101} {\bibfield  {journal} {\bibinfo
  {journal} {Phys. Rev. A}\ }\textbf {\bibinfo {volume} {96}},\ \bibinfo
  {pages} {012101} (\bibinfo {year} {2017})}\BibitemShut {NoStop}%
\bibitem [{\citenamefont {Fonseca}\ \emph {et~al.}(2018)\citenamefont
  {Fonseca}, \citenamefont {de~Rosier}, \citenamefont {V\'ertesi},
  \citenamefont {Laskowski},\ and\ \citenamefont {Parisio}}]{Fonseca_2018}%
  \BibitemOpen
  \bibfield  {author} {\bibinfo {author} {\bibfnamefont {A.}~\bibnamefont
  {Fonseca}}, \bibinfo {author} {\bibfnamefont {A.}~\bibnamefont {de~Rosier}},
  \bibinfo {author} {\bibfnamefont {T.}~\bibnamefont {V\'ertesi}}, \bibinfo
  {author} {\bibfnamefont {W.}~\bibnamefont {Laskowski}},\ and\ \bibinfo
  {author} {\bibfnamefont {F.}~\bibnamefont {Parisio}},\ }\bibfield  {title}
  {\bibinfo {title} {Survey on the {B}ell nonlocality of a pair of entangled
  qudits},\ }\href {https://doi.org/10.1103/PhysRevA.98.042105} {\bibfield
  {journal} {\bibinfo  {journal} {Phys. Rev. A}\ }\textbf {\bibinfo {volume}
  {98}},\ \bibinfo {pages} {042105} (\bibinfo {year} {2018})}\BibitemShut
  {NoStop}%
\bibitem [{\citenamefont {Laing}\ \emph {et~al.}(2010)\citenamefont {Laing},
  \citenamefont {Scarani}, \citenamefont {Rarity},\ and\ \citenamefont
  {O'Brien}}]{Laing_2010}%
  \BibitemOpen
  \bibfield  {author} {\bibinfo {author} {\bibfnamefont {A.}~\bibnamefont
  {Laing}}, \bibinfo {author} {\bibfnamefont {V.}~\bibnamefont {Scarani}},
  \bibinfo {author} {\bibfnamefont {J.~G.}\ \bibnamefont {Rarity}},\ and\
  \bibinfo {author} {\bibfnamefont {J.~L.}\ \bibnamefont {O'Brien}},\
  }\bibfield  {title} {\bibinfo {title} {Reference-frame-independent quantum
  key distribution},\ }\href {https://doi.org/10.1103/PhysRevA.82.012304}
  {\bibfield  {journal} {\bibinfo  {journal} {Phys. Rev. A}\ }\textbf {\bibinfo
  {volume} {82}},\ \bibinfo {pages} {012304} (\bibinfo {year}
  {2010})}\BibitemShut {NoStop}%
\bibitem [{\citenamefont {Wabnig}\ \emph {et~al.}(2013)\citenamefont {Wabnig},
  \citenamefont {Bitauld}, \citenamefont {Li}, \citenamefont {Laing},
  \citenamefont {O'Brien},\ and\ \citenamefont {Niskanen}}]{Wabnig_2013}%
  \BibitemOpen
  \bibfield  {author} {\bibinfo {author} {\bibfnamefont {J.}~\bibnamefont
  {Wabnig}}, \bibinfo {author} {\bibfnamefont {D.}~\bibnamefont {Bitauld}},
  \bibinfo {author} {\bibfnamefont {H.~W.}\ \bibnamefont {Li}}, \bibinfo
  {author} {\bibfnamefont {A.}~\bibnamefont {Laing}}, \bibinfo {author}
  {\bibfnamefont {J.~L.}\ \bibnamefont {O'Brien}},\ and\ \bibinfo {author}
  {\bibfnamefont {A.~O.}\ \bibnamefont {Niskanen}},\ }\bibfield  {title}
  {\bibinfo {title} {Demonstration of free-space reference frame independent
  quantum key distribution},\ }\href
  {https://doi.org/10.1088/1367-2630/15/7/073001} {\bibfield  {journal}
  {\bibinfo  {journal} {New J. Phys.}\ }\textbf {\bibinfo {volume} {15}},\
  \bibinfo {pages} {073001} (\bibinfo {year} {2013})}\BibitemShut {NoStop}%
\bibitem [{\citenamefont {Slater}\ \emph {et~al.}(2014)\citenamefont {Slater},
  \citenamefont {Branciard}, \citenamefont {Brunner},\ and\ \citenamefont
  {Tittel}}]{Slater_2014}%
  \BibitemOpen
  \bibfield  {author} {\bibinfo {author} {\bibfnamefont {J.~A.}\ \bibnamefont
  {Slater}}, \bibinfo {author} {\bibfnamefont {C.}~\bibnamefont {Branciard}},
  \bibinfo {author} {\bibfnamefont {N.}~\bibnamefont {Brunner}},\ and\ \bibinfo
  {author} {\bibfnamefont {W.}~\bibnamefont {Tittel}},\ }\bibfield  {title}
  {\bibinfo {title} {Device-dependent and device-independent quantum key
  distribution without a shared reference frame},\ }\href
  {https://doi.org/10.1088/1367-2630/16/4/043002} {\bibfield  {journal}
  {\bibinfo  {journal} {New J. Phys.}\ }\textbf {\bibinfo {volume} {16}},\
  \bibinfo {pages} {043002} (\bibinfo {year} {2014})}\BibitemShut {NoStop}%
\bibitem [{\citenamefont {Laskowski}\ \emph {et~al.}(2012)\citenamefont
  {Laskowski}, \citenamefont {Richart}, \citenamefont {Schwemmer},
  \citenamefont {Paterek},\ and\ \citenamefont {Weinfurter}}]{Laskowski_2012}%
  \BibitemOpen
  \bibfield  {author} {\bibinfo {author} {\bibfnamefont {W.}~\bibnamefont
  {Laskowski}}, \bibinfo {author} {\bibfnamefont {D.}~\bibnamefont {Richart}},
  \bibinfo {author} {\bibfnamefont {C.}~\bibnamefont {Schwemmer}}, \bibinfo
  {author} {\bibfnamefont {T.}~\bibnamefont {Paterek}},\ and\ \bibinfo {author}
  {\bibfnamefont {H.}~\bibnamefont {Weinfurter}},\ }\bibfield  {title}
  {\bibinfo {title} {Experimental schmidt decomposition and state independent
  entanglement detection},\ }\href
  {https://doi.org/10.1103/PhysRevLett.108.240501} {\bibfield  {journal}
  {\bibinfo  {journal} {Phys. Rev. Lett.}\ }\textbf {\bibinfo {volume} {108}},\
  \bibinfo {pages} {240501} (\bibinfo {year} {2012})}\BibitemShut {NoStop}%
\bibitem [{\citenamefont {Laskowski}\ \emph {et~al.}(2013)\citenamefont
  {Laskowski}, \citenamefont {Schwemmer}, \citenamefont {Richart},
  \citenamefont {Knips}, \citenamefont {Paterek},\ and\ \citenamefont
  {Weinfurter}}]{Laskowski_2013}%
  \BibitemOpen
  \bibfield  {author} {\bibinfo {author} {\bibfnamefont {W.}~\bibnamefont
  {Laskowski}}, \bibinfo {author} {\bibfnamefont {C.}~\bibnamefont
  {Schwemmer}}, \bibinfo {author} {\bibfnamefont {D.}~\bibnamefont {Richart}},
  \bibinfo {author} {\bibfnamefont {L.}~\bibnamefont {Knips}}, \bibinfo
  {author} {\bibfnamefont {T.}~\bibnamefont {Paterek}},\ and\ \bibinfo {author}
  {\bibfnamefont {H.}~\bibnamefont {Weinfurter}},\ }\bibfield  {title}
  {\bibinfo {title} {Optimized state-independent entanglement detection based
  on a geometrical threshold criterion},\ }\href
  {https://doi.org/10.1103/PhysRevA.88.022327} {\bibfield  {journal} {\bibinfo
  {journal} {Phys. Rev. A}\ }\textbf {\bibinfo {volume} {88}},\ \bibinfo
  {pages} {022327} (\bibinfo {year} {2013})}\BibitemShut {NoStop}%
\bibitem [{\citenamefont {Knips}\ \emph {et~al.}(2020)\citenamefont {Knips},
  \citenamefont {Dziewior}, \citenamefont {Kłobus}, \citenamefont {Laskowski},
  \citenamefont {Paterek}, \citenamefont {Shadbolt}, \citenamefont
  {Weinfurter},\ and\ \citenamefont {Meinecke}}]{Knips_2020}%
  \BibitemOpen
  \bibfield  {author} {\bibinfo {author} {\bibfnamefont {L.}~\bibnamefont
  {Knips}}, \bibinfo {author} {\bibfnamefont {J.}~\bibnamefont {Dziewior}},
  \bibinfo {author} {\bibfnamefont {W.}~\bibnamefont {Kłobus}}, \bibinfo
  {author} {\bibfnamefont {W.}~\bibnamefont {Laskowski}}, \bibinfo {author}
  {\bibfnamefont {T.}~\bibnamefont {Paterek}}, \bibinfo {author} {\bibfnamefont
  {P.~J.}\ \bibnamefont {Shadbolt}}, \bibinfo {author} {\bibfnamefont
  {H.}~\bibnamefont {Weinfurter}},\ and\ \bibinfo {author} {\bibfnamefont
  {J.~D.~A.}\ \bibnamefont {Meinecke}},\ }\bibfield  {title} {\bibinfo {title}
  {Multipartite entanglement analysis from random correlations},\ }\href
  {https://doi.org/10.1038/s41534-020-0281-5} {\bibfield  {journal} {\bibinfo
  {journal} {npj Quantum Inf.}\ }\textbf {\bibinfo {volume} {6}},\ \bibinfo
  {pages} {51} (\bibinfo {year} {2020})}\BibitemShut {NoStop}%
\bibitem [{\citenamefont {Wyderka}\ \emph {et~al.}(2023)\citenamefont
  {Wyderka}, \citenamefont {Ketterer}, \citenamefont {Imai}, \citenamefont
  {B\"onsel}, \citenamefont {Jones}, \citenamefont {Kirby}, \citenamefont
  {Yu},\ and\ \citenamefont {G\"uhne}}]{Wyderka_2023}%
  \BibitemOpen
  \bibfield  {author} {\bibinfo {author} {\bibfnamefont {N.}~\bibnamefont
  {Wyderka}}, \bibinfo {author} {\bibfnamefont {A.}~\bibnamefont {Ketterer}},
  \bibinfo {author} {\bibfnamefont {S.}~\bibnamefont {Imai}}, \bibinfo {author}
  {\bibfnamefont {J.~L.}\ \bibnamefont {B\"onsel}}, \bibinfo {author}
  {\bibfnamefont {D.~E.}\ \bibnamefont {Jones}}, \bibinfo {author}
  {\bibfnamefont {B.~T.}\ \bibnamefont {Kirby}}, \bibinfo {author}
  {\bibfnamefont {X.-D.}\ \bibnamefont {Yu}},\ and\ \bibinfo {author}
  {\bibfnamefont {O.}~\bibnamefont {G\"uhne}},\ }\bibfield  {title} {\bibinfo
  {title} {Complete characterization of quantum correlations by randomized
  measurements},\ }\href {https://doi.org/10.1103/PhysRevLett.131.090201}
  {\bibfield  {journal} {\bibinfo  {journal} {Phys. Rev. Lett.}\ }\textbf
  {\bibinfo {volume} {131}},\ \bibinfo {pages} {090201} (\bibinfo {year}
  {2023})}\BibitemShut {NoStop}%
\bibitem [{\citenamefont {Cieśliński}\ \emph {et~al.}(2024)\citenamefont
  {Cieśliński}, \citenamefont {Imai}, \citenamefont {Dziewior}, \citenamefont
  {Gühne}, \citenamefont {Knips}, \citenamefont {Laskowski}, \citenamefont
  {Meinecke}, \citenamefont {Paterek},\ and\ \citenamefont
  {Vértesi}}]{Cieslinski_2024}%
  \BibitemOpen
  \bibfield  {author} {\bibinfo {author} {\bibfnamefont {P.}~\bibnamefont
  {Cieśliński}}, \bibinfo {author} {\bibfnamefont {S.}~\bibnamefont {Imai}},
  \bibinfo {author} {\bibfnamefont {J.}~\bibnamefont {Dziewior}}, \bibinfo
  {author} {\bibfnamefont {O.}~\bibnamefont {Gühne}}, \bibinfo {author}
  {\bibfnamefont {L.}~\bibnamefont {Knips}}, \bibinfo {author} {\bibfnamefont
  {W.}~\bibnamefont {Laskowski}}, \bibinfo {author} {\bibfnamefont
  {J.}~\bibnamefont {Meinecke}}, \bibinfo {author} {\bibfnamefont
  {T.}~\bibnamefont {Paterek}},\ and\ \bibinfo {author} {\bibfnamefont
  {T.}~\bibnamefont {Vértesi}},\ }\bibfield  {title} {\bibinfo {title}
  {Analysing quantum systems with randomised measurements},\ }\href
  {https://doi.org/https://doi.org/10.1016/j.physrep.2024.09.009} {\bibfield
  {journal} {\bibinfo  {journal} {Phys. Rep.}\ }\textbf {\bibinfo {volume}
  {1095}},\ \bibinfo {pages} {1} (\bibinfo {year} {2024})}\BibitemShut
  {NoStop}%
\bibitem [{\citenamefont {Skinner}\ \emph {et~al.}(2019)\citenamefont
  {Skinner}, \citenamefont {Ruhman},\ and\ \citenamefont
  {Nahum}}]{Skinner_2019}%
  \BibitemOpen
  \bibfield  {author} {\bibinfo {author} {\bibfnamefont {B.}~\bibnamefont
  {Skinner}}, \bibinfo {author} {\bibfnamefont {J.}~\bibnamefont {Ruhman}},\
  and\ \bibinfo {author} {\bibfnamefont {A.}~\bibnamefont {Nahum}},\ }\bibfield
   {title} {\bibinfo {title} {Measurement-induced phase transitions in the
  dynamics of entanglement},\ }\href
  {https://doi.org/10.1103/PhysRevX.9.031009} {\bibfield  {journal} {\bibinfo
  {journal} {Phys. Rev. X}\ }\textbf {\bibinfo {volume} {9}},\ \bibinfo {pages}
  {031009} (\bibinfo {year} {2019})}\BibitemShut {NoStop}%
\bibitem [{\citenamefont {Noel}\ \emph {et~al.}(2022)\citenamefont {Noel},
  \citenamefont {Niroula}, \citenamefont {Zhu}, \citenamefont {Risinger},
  \citenamefont {Egan}, \citenamefont {Biswas}, \citenamefont {Cetina},
  \citenamefont {Gorshkov}, \citenamefont {Gullans}, \citenamefont {Huse},\
  and\ \citenamefont {Monroe}}]{Noel_2022}%
  \BibitemOpen
  \bibfield  {author} {\bibinfo {author} {\bibfnamefont {C.}~\bibnamefont
  {Noel}}, \bibinfo {author} {\bibfnamefont {P.}~\bibnamefont {Niroula}},
  \bibinfo {author} {\bibfnamefont {D.}~\bibnamefont {Zhu}}, \bibinfo {author}
  {\bibfnamefont {A.}~\bibnamefont {Risinger}}, \bibinfo {author}
  {\bibfnamefont {L.}~\bibnamefont {Egan}}, \bibinfo {author} {\bibfnamefont
  {D.}~\bibnamefont {Biswas}}, \bibinfo {author} {\bibfnamefont
  {M.}~\bibnamefont {Cetina}}, \bibinfo {author} {\bibfnamefont {A.~V.}\
  \bibnamefont {Gorshkov}}, \bibinfo {author} {\bibfnamefont {M.~J.}\
  \bibnamefont {Gullans}}, \bibinfo {author} {\bibfnamefont {D.~A.}\
  \bibnamefont {Huse}},\ and\ \bibinfo {author} {\bibfnamefont
  {C.}~\bibnamefont {Monroe}},\ }\bibfield  {title} {\bibinfo {title}
  {Measurement-induced quantum phases realized in a trapped-ion quantum
  computer},\ }\href {https://doi.org/10.1038/s41567-022-01619-7} {\bibfield
  {journal} {\bibinfo  {journal} {Nat. Phys.}\ }\textbf {\bibinfo {volume}
  {18}},\ \bibinfo {pages} {760} (\bibinfo {year} {2022})}\BibitemShut
  {NoStop}%
\bibitem [{\citenamefont {Koh}\ \emph {et~al.}(2023)\citenamefont {Koh},
  \citenamefont {Sun}, \citenamefont {Motta},\ and\ \citenamefont
  {Minnich}}]{Koh_2023}%
  \BibitemOpen
  \bibfield  {author} {\bibinfo {author} {\bibfnamefont {J.~M.}\ \bibnamefont
  {Koh}}, \bibinfo {author} {\bibfnamefont {S.-N.}\ \bibnamefont {Sun}},
  \bibinfo {author} {\bibfnamefont {M.}~\bibnamefont {Motta}},\ and\ \bibinfo
  {author} {\bibfnamefont {A.~J.}\ \bibnamefont {Minnich}},\ }\bibfield
  {title} {\bibinfo {title} {Measurement-induced entanglement phase transition
  on a superconducting quantum processor with mid-circuit readout},\ }\href
  {https://doi.org/10.1038/s41567-023-02076-6} {\bibfield  {journal} {\bibinfo
  {journal} {Nat. Phys.}\ }\textbf {\bibinfo {volume} {19}},\ \bibinfo {pages}
  {1314} (\bibinfo {year} {2023})}\BibitemShut {NoStop}%
\bibitem [{\citenamefont {Morvan}\ \emph {et~al.}(2024)\citenamefont {Morvan}
  \emph {et~al.}}]{Morvan_2024}%
  \BibitemOpen
  \bibfield  {author} {\bibinfo {author} {\bibfnamefont {A.}~\bibnamefont
  {Morvan}} \emph {et~al.},\ }\bibfield  {title} {\bibinfo {title} {Phase
  transitions in random circuit sampling},\ }\href
  {https://doi.org/10.1038/s41586-024-07998-6} {\bibfield  {journal} {\bibinfo
  {journal} {Nature}\ }\textbf {\bibinfo {volume} {634}},\ \bibinfo {pages}
  {328} (\bibinfo {year} {2024})}\BibitemShut {NoStop}%
\bibitem [{\citenamefont {Elben}\ \emph {et~al.}(2019)\citenamefont {Elben},
  \citenamefont {Vermersch}, \citenamefont {Roos},\ and\ \citenamefont
  {Zoller}}]{Elben_2019}%
  \BibitemOpen
  \bibfield  {author} {\bibinfo {author} {\bibfnamefont {A.}~\bibnamefont
  {Elben}}, \bibinfo {author} {\bibfnamefont {B.}~\bibnamefont {Vermersch}},
  \bibinfo {author} {\bibfnamefont {C.~F.}\ \bibnamefont {Roos}},\ and\
  \bibinfo {author} {\bibfnamefont {P.}~\bibnamefont {Zoller}},\ }\bibfield
  {title} {\bibinfo {title} {Statistical correlations between locally
  randomized measurements: A toolbox for probing entanglement in many-body
  quantum states},\ }\href {https://doi.org/10.1103/PhysRevA.99.052323}
  {\bibfield  {journal} {\bibinfo  {journal} {Phys. Rev. A}\ }\textbf {\bibinfo
  {volume} {99}},\ \bibinfo {pages} {052323} (\bibinfo {year}
  {2019})}\BibitemShut {NoStop}%
\bibitem [{\citenamefont {Elben}\ \emph
  {et~al.}(2020{\natexlab{a}})\citenamefont {Elben}, \citenamefont {Yu},
  \citenamefont {Zhu}, \citenamefont {Hafezi}, \citenamefont {Pollmann},
  \citenamefont {Zoller},\ and\ \citenamefont {Vermersch}}]{Elben_2020}%
  \BibitemOpen
  \bibfield  {author} {\bibinfo {author} {\bibfnamefont {A.}~\bibnamefont
  {Elben}}, \bibinfo {author} {\bibfnamefont {J.}~\bibnamefont {Yu}}, \bibinfo
  {author} {\bibfnamefont {G.}~\bibnamefont {Zhu}}, \bibinfo {author}
  {\bibfnamefont {M.}~\bibnamefont {Hafezi}}, \bibinfo {author} {\bibfnamefont
  {F.}~\bibnamefont {Pollmann}}, \bibinfo {author} {\bibfnamefont
  {P.}~\bibnamefont {Zoller}},\ and\ \bibinfo {author} {\bibfnamefont
  {B.}~\bibnamefont {Vermersch}},\ }\bibfield  {title} {\bibinfo {title}
  {Many-body topological invariants from randomized measurements in synthetic
  quantum matter},\ }\href {https://doi.org/10.1126/sciadv.aaz3666} {\bibfield
  {journal} {\bibinfo  {journal} {Sci. Adv.}\ }\textbf {\bibinfo {volume}
  {6}},\ \bibinfo {pages} {eaaz3666} (\bibinfo {year}
  {2020}{\natexlab{a}})}\BibitemShut {NoStop}%
\bibitem [{\citenamefont {Elben}\ \emph
  {et~al.}(2020{\natexlab{b}})\citenamefont {Elben}, \citenamefont {Kueng},
  \citenamefont {Huang}, \citenamefont {van Bijnen}, \citenamefont {Kokail},
  \citenamefont {Dalmonte}, \citenamefont {Calabrese}, \citenamefont {Kraus},
  \citenamefont {Preskill}, \citenamefont {Zoller},\ and\ \citenamefont
  {Vermersch}}]{Elben_2020_2}%
  \BibitemOpen
  \bibfield  {author} {\bibinfo {author} {\bibfnamefont {A.}~\bibnamefont
  {Elben}}, \bibinfo {author} {\bibfnamefont {R.}~\bibnamefont {Kueng}},
  \bibinfo {author} {\bibfnamefont {H.-Y.~R.}\ \bibnamefont {Huang}}, \bibinfo
  {author} {\bibfnamefont {R.}~\bibnamefont {van Bijnen}}, \bibinfo {author}
  {\bibfnamefont {C.}~\bibnamefont {Kokail}}, \bibinfo {author} {\bibfnamefont
  {M.}~\bibnamefont {Dalmonte}}, \bibinfo {author} {\bibfnamefont
  {P.}~\bibnamefont {Calabrese}}, \bibinfo {author} {\bibfnamefont
  {B.}~\bibnamefont {Kraus}}, \bibinfo {author} {\bibfnamefont
  {J.}~\bibnamefont {Preskill}}, \bibinfo {author} {\bibfnamefont
  {P.}~\bibnamefont {Zoller}},\ and\ \bibinfo {author} {\bibfnamefont
  {B.}~\bibnamefont {Vermersch}},\ }\bibfield  {title} {\bibinfo {title}
  {Mixed-state entanglement from local randomized measurements},\ }\href
  {https://doi.org/10.1103/PhysRevLett.125.200501} {\bibfield  {journal}
  {\bibinfo  {journal} {Phys. Rev. Lett.}\ }\textbf {\bibinfo {volume} {125}},\
  \bibinfo {pages} {200501} (\bibinfo {year} {2020}{\natexlab{b}})}\BibitemShut
  {NoStop}%
\bibitem [{\citenamefont {Notarnicola}\ \emph {et~al.}(2023)\citenamefont
  {Notarnicola}, \citenamefont {Elben}, \citenamefont {Lahaye}, \citenamefont
  {Browaeys}, \citenamefont {Montangero},\ and\ \citenamefont
  {Vermersch}}]{Notarnicola_2023}%
  \BibitemOpen
  \bibfield  {author} {\bibinfo {author} {\bibfnamefont {S.}~\bibnamefont
  {Notarnicola}}, \bibinfo {author} {\bibfnamefont {A.}~\bibnamefont {Elben}},
  \bibinfo {author} {\bibfnamefont {T.}~\bibnamefont {Lahaye}}, \bibinfo
  {author} {\bibfnamefont {A.}~\bibnamefont {Browaeys}}, \bibinfo {author}
  {\bibfnamefont {S.}~\bibnamefont {Montangero}},\ and\ \bibinfo {author}
  {\bibfnamefont {B.}~\bibnamefont {Vermersch}},\ }\bibfield  {title} {\bibinfo
  {title} {A randomized measurement toolbox for an interacting rydberg-atom
  quantum simulator},\ }\href {https://doi.org/10.1088/1367-2630/acfcd3}
  {\bibfield  {journal} {\bibinfo  {journal} {New J. Phys.}\ }\textbf {\bibinfo
  {volume} {25}},\ \bibinfo {pages} {103006} (\bibinfo {year}
  {2023})}\BibitemShut {NoStop}%
\bibitem [{\citenamefont {Knill}\ \emph {et~al.}(2008)\citenamefont {Knill},
  \citenamefont {Leibfried}, \citenamefont {Reichle}, \citenamefont {Britton},
  \citenamefont {Blakestad}, \citenamefont {Jost}, \citenamefont {Langer},
  \citenamefont {Ozeri}, \citenamefont {Seidelin},\ and\ \citenamefont
  {Wineland}}]{Knill_2008}%
  \BibitemOpen
  \bibfield  {author} {\bibinfo {author} {\bibfnamefont {E.}~\bibnamefont
  {Knill}}, \bibinfo {author} {\bibfnamefont {D.}~\bibnamefont {Leibfried}},
  \bibinfo {author} {\bibfnamefont {R.}~\bibnamefont {Reichle}}, \bibinfo
  {author} {\bibfnamefont {J.}~\bibnamefont {Britton}}, \bibinfo {author}
  {\bibfnamefont {R.~B.}\ \bibnamefont {Blakestad}}, \bibinfo {author}
  {\bibfnamefont {J.~D.}\ \bibnamefont {Jost}}, \bibinfo {author}
  {\bibfnamefont {C.}~\bibnamefont {Langer}}, \bibinfo {author} {\bibfnamefont
  {R.}~\bibnamefont {Ozeri}}, \bibinfo {author} {\bibfnamefont
  {S.}~\bibnamefont {Seidelin}},\ and\ \bibinfo {author} {\bibfnamefont
  {D.~J.}\ \bibnamefont {Wineland}},\ }\bibfield  {title} {\bibinfo {title}
  {Randomized benchmarking of quantum gates},\ }\href
  {https://doi.org/10.1103/PhysRevA.77.012307} {\bibfield  {journal} {\bibinfo
  {journal} {Phys. Rev. A}\ }\textbf {\bibinfo {volume} {77}},\ \bibinfo
  {pages} {012307} (\bibinfo {year} {2008})}\BibitemShut {NoStop}%
\bibitem [{\citenamefont {Elben}\ \emph
  {et~al.}(2020{\natexlab{c}})\citenamefont {Elben}, \citenamefont {Vermersch},
  \citenamefont {van Bijnen}, \citenamefont {Kokail}, \citenamefont {Brydges},
  \citenamefont {Maier}, \citenamefont {Joshi}, \citenamefont {Blatt},
  \citenamefont {Roos},\ and\ \citenamefont {Zoller}}]{Elben_2020_3}%
  \BibitemOpen
  \bibfield  {author} {\bibinfo {author} {\bibfnamefont {A.}~\bibnamefont
  {Elben}}, \bibinfo {author} {\bibfnamefont {B.}~\bibnamefont {Vermersch}},
  \bibinfo {author} {\bibfnamefont {R.}~\bibnamefont {van Bijnen}}, \bibinfo
  {author} {\bibfnamefont {C.}~\bibnamefont {Kokail}}, \bibinfo {author}
  {\bibfnamefont {T.}~\bibnamefont {Brydges}}, \bibinfo {author} {\bibfnamefont
  {C.}~\bibnamefont {Maier}}, \bibinfo {author} {\bibfnamefont {M.~K.}\
  \bibnamefont {Joshi}}, \bibinfo {author} {\bibfnamefont {R.}~\bibnamefont
  {Blatt}}, \bibinfo {author} {\bibfnamefont {C.~F.}\ \bibnamefont {Roos}},\
  and\ \bibinfo {author} {\bibfnamefont {P.}~\bibnamefont {Zoller}},\
  }\bibfield  {title} {\bibinfo {title} {Cross-platform verification of
  intermediate scale quantum devices},\ }\href
  {https://doi.org/10.1103/PhysRevLett.124.010504} {\bibfield  {journal}
  {\bibinfo  {journal} {Phys. Rev. Lett.}\ }\textbf {\bibinfo {volume} {124}},\
  \bibinfo {pages} {010504} (\bibinfo {year} {2020}{\natexlab{c}})}\BibitemShut
  {NoStop}%
\bibitem [{\citenamefont {Helsen}\ \emph {et~al.}(2022)\citenamefont {Helsen},
  \citenamefont {Roth}, \citenamefont {Onorati}, \citenamefont {Werner},\ and\
  \citenamefont {Eisert}}]{Helsen_2022}%
  \BibitemOpen
  \bibfield  {author} {\bibinfo {author} {\bibfnamefont {J.}~\bibnamefont
  {Helsen}}, \bibinfo {author} {\bibfnamefont {I.}~\bibnamefont {Roth}},
  \bibinfo {author} {\bibfnamefont {E.}~\bibnamefont {Onorati}}, \bibinfo
  {author} {\bibfnamefont {A.}~\bibnamefont {Werner}},\ and\ \bibinfo {author}
  {\bibfnamefont {J.}~\bibnamefont {Eisert}},\ }\bibfield  {title} {\bibinfo
  {title} {General framework for randomized benchmarking},\ }\href
  {https://doi.org/10.1103/PRXQuantum.3.020357} {\bibfield  {journal} {\bibinfo
   {journal} {PRX Quantum}\ }\textbf {\bibinfo {volume} {3}},\ \bibinfo {pages}
  {020357} (\bibinfo {year} {2022})}\BibitemShut {NoStop}%
\bibitem [{\citenamefont {Elben}\ \emph {et~al.}(2018)\citenamefont {Elben},
  \citenamefont {Vermersch}, \citenamefont {Dalmonte}, \citenamefont {Cirac},\
  and\ \citenamefont {Zoller}}]{Elben_2018}%
  \BibitemOpen
  \bibfield  {author} {\bibinfo {author} {\bibfnamefont {A.}~\bibnamefont
  {Elben}}, \bibinfo {author} {\bibfnamefont {B.}~\bibnamefont {Vermersch}},
  \bibinfo {author} {\bibfnamefont {M.}~\bibnamefont {Dalmonte}}, \bibinfo
  {author} {\bibfnamefont {J.~I.}\ \bibnamefont {Cirac}},\ and\ \bibinfo
  {author} {\bibfnamefont {P.}~\bibnamefont {Zoller}},\ }\bibfield  {title}
  {\bibinfo {title} {R\'enyi entropies from random quenches in atomic hubbard
  and spin models},\ }\href {https://doi.org/10.1103/PhysRevLett.120.050406}
  {\bibfield  {journal} {\bibinfo  {journal} {Phys. Rev. Lett.}\ }\textbf
  {\bibinfo {volume} {120}},\ \bibinfo {pages} {050406} (\bibinfo {year}
  {2018})}\BibitemShut {NoStop}%
\bibitem [{\citenamefont {Brydges}\ \emph {et~al.}(2019)\citenamefont
  {Brydges}, \citenamefont {Elben}, \citenamefont {Jurcevic}, \citenamefont
  {Vermersch}, \citenamefont {Maier}, \citenamefont {Lanyon}, \citenamefont
  {Zoller}, \citenamefont {Blatt},\ and\ \citenamefont {Roos}}]{Brydges_2019}%
  \BibitemOpen
  \bibfield  {author} {\bibinfo {author} {\bibfnamefont {T.}~\bibnamefont
  {Brydges}}, \bibinfo {author} {\bibfnamefont {A.}~\bibnamefont {Elben}},
  \bibinfo {author} {\bibfnamefont {P.}~\bibnamefont {Jurcevic}}, \bibinfo
  {author} {\bibfnamefont {B.}~\bibnamefont {Vermersch}}, \bibinfo {author}
  {\bibfnamefont {C.}~\bibnamefont {Maier}}, \bibinfo {author} {\bibfnamefont
  {B.~P.}\ \bibnamefont {Lanyon}}, \bibinfo {author} {\bibfnamefont
  {P.}~\bibnamefont {Zoller}}, \bibinfo {author} {\bibfnamefont
  {R.}~\bibnamefont {Blatt}},\ and\ \bibinfo {author} {\bibfnamefont {C.~F.}\
  \bibnamefont {Roos}},\ }\bibfield  {title} {\bibinfo {title} {Probing rényi
  entanglement entropy via randomized measurements},\ }\href
  {https://doi.org/10.1126/science.aau4963} {\bibfield  {journal} {\bibinfo
  {journal} {Science}\ }\textbf {\bibinfo {volume} {364}},\ \bibinfo {pages}
  {260} (\bibinfo {year} {2019})}\BibitemShut {NoStop}%
\bibitem [{\citenamefont {Huang}\ \emph {et~al.}(2020)\citenamefont {Huang},
  \citenamefont {Kueng},\ and\ \citenamefont {Preskill}}]{Huang_2020}%
  \BibitemOpen
  \bibfield  {author} {\bibinfo {author} {\bibfnamefont {H.-Y.}\ \bibnamefont
  {Huang}}, \bibinfo {author} {\bibfnamefont {R.}~\bibnamefont {Kueng}},\ and\
  \bibinfo {author} {\bibfnamefont {J.}~\bibnamefont {Preskill}},\ }\bibfield
  {title} {\bibinfo {title} {Predicting many properties of a quantum system
  from very few measurements},\ }\href
  {https://doi.org/10.1038/s41567-020-0932-7} {\bibfield  {journal} {\bibinfo
  {journal} {Nat. Phys.}\ }\textbf {\bibinfo {volume} {16}},\ \bibinfo {pages}
  {1050} (\bibinfo {year} {2020})}\BibitemShut {NoStop}%
\bibitem [{\citenamefont {Knips}(2020)}]{Knips_2020_2}%
  \BibitemOpen
  \bibfield  {author} {\bibinfo {author} {\bibfnamefont {L.}~\bibnamefont
  {Knips}},\ }\bibfield  {title} {\bibinfo {title} {A {M}oment for {R}andom
  {M}easurements},\ }\href {https://doi.org/10.22331/qv-2020-11-19-47}
  {\bibfield  {journal} {\bibinfo  {journal} {{Quantum Views}}\ }\textbf
  {\bibinfo {volume} {4}},\ \bibinfo {pages} {47} (\bibinfo {year}
  {2020})}\BibitemShut {NoStop}%
\bibitem [{\citenamefont {Tran}\ \emph {et~al.}(2015)\citenamefont {Tran},
  \citenamefont {Daki\ifmmode~\acute{c}\else \'{c}\fi{}}, \citenamefont
  {Arnault}, \citenamefont {Laskowski},\ and\ \citenamefont
  {Paterek}}]{Tran_2015}%
  \BibitemOpen
  \bibfield  {author} {\bibinfo {author} {\bibfnamefont {M.~C.}\ \bibnamefont
  {Tran}}, \bibinfo {author} {\bibfnamefont {B.}~\bibnamefont
  {Daki\ifmmode~\acute{c}\else \'{c}\fi{}}}, \bibinfo {author} {\bibfnamefont
  {F.~m.~c.}\ \bibnamefont {Arnault}}, \bibinfo {author} {\bibfnamefont
  {W.}~\bibnamefont {Laskowski}},\ and\ \bibinfo {author} {\bibfnamefont
  {T.}~\bibnamefont {Paterek}},\ }\bibfield  {title} {\bibinfo {title} {Quantum
  entanglement from random measurements},\ }\href
  {https://doi.org/10.1103/PhysRevA.92.050301} {\bibfield  {journal} {\bibinfo
  {journal} {Phys. Rev. A}\ }\textbf {\bibinfo {volume} {92}},\ \bibinfo
  {pages} {050301} (\bibinfo {year} {2015})}\BibitemShut {NoStop}%
\bibitem [{\citenamefont {Tran}\ \emph {et~al.}(2016)\citenamefont {Tran},
  \citenamefont {Daki\ifmmode~\acute{c}\else \'{c}\fi{}}, \citenamefont
  {Laskowski},\ and\ \citenamefont {Paterek}}]{Tran_2016}%
  \BibitemOpen
  \bibfield  {author} {\bibinfo {author} {\bibfnamefont {M.~C.}\ \bibnamefont
  {Tran}}, \bibinfo {author} {\bibfnamefont {B.}~\bibnamefont
  {Daki\ifmmode~\acute{c}\else \'{c}\fi{}}}, \bibinfo {author} {\bibfnamefont
  {W.}~\bibnamefont {Laskowski}},\ and\ \bibinfo {author} {\bibfnamefont
  {T.}~\bibnamefont {Paterek}},\ }\bibfield  {title} {\bibinfo {title}
  {Correlations between outcomes of random measurements},\ }\href
  {https://doi.org/10.1103/PhysRevA.94.042302} {\bibfield  {journal} {\bibinfo
  {journal} {Phys. Rev. A}\ }\textbf {\bibinfo {volume} {94}},\ \bibinfo
  {pages} {042302} (\bibinfo {year} {2016})}\BibitemShut {NoStop}%
\bibitem [{\citenamefont {Shravan}\ \emph {et~al.}(2024)\citenamefont
  {Shravan}, \citenamefont {Morelli}, \citenamefont {G\"uhne},\ and\
  \citenamefont {Imai}}]{Shravan_2024}%
  \BibitemOpen
  \bibfield  {author} {\bibinfo {author} {\bibfnamefont {S.}~\bibnamefont
  {Shravan}}, \bibinfo {author} {\bibfnamefont {S.}~\bibnamefont {Morelli}},
  \bibinfo {author} {\bibfnamefont {O.}~\bibnamefont {G\"uhne}},\ and\ \bibinfo
  {author} {\bibfnamefont {S.}~\bibnamefont {Imai}},\ }\bibfield  {title}
  {\bibinfo {title} {Geometry of two-body correlations in three-qubit states},\
  }\href {https://doi.org/10.1103/PhysRevA.110.062419} {\bibfield  {journal}
  {\bibinfo  {journal} {Phys. Rev. A}\ }\textbf {\bibinfo {volume} {110}},\
  \bibinfo {pages} {062419} (\bibinfo {year} {2024})}\BibitemShut {NoStop}%
\bibitem [{\citenamefont {Ketterer}\ \emph {et~al.}(2020)\citenamefont
  {Ketterer}, \citenamefont {Wyderka},\ and\ \citenamefont
  {G{\"{u}}hne}}]{Ketterer_2020}%
  \BibitemOpen
  \bibfield  {author} {\bibinfo {author} {\bibfnamefont {A.}~\bibnamefont
  {Ketterer}}, \bibinfo {author} {\bibfnamefont {N.}~\bibnamefont {Wyderka}},\
  and\ \bibinfo {author} {\bibfnamefont {O.}~\bibnamefont {G{\"{u}}hne}},\
  }\bibfield  {title} {\bibinfo {title} {Entanglement characterization using
  quantum designs},\ }\href {https://doi.org/10.22331/q-2020-09-16-325}
  {\bibfield  {journal} {\bibinfo  {journal} {{Quantum}}\ }\textbf {\bibinfo
  {volume} {4}},\ \bibinfo {pages} {325} (\bibinfo {year} {2020})}\BibitemShut
  {NoStop}%
\bibitem [{\citenamefont {Ketterer}\ \emph {et~al.}(2019)\citenamefont
  {Ketterer}, \citenamefont {Wyderka},\ and\ \citenamefont
  {G\"uhne}}]{Ketterer_2018}%
  \BibitemOpen
  \bibfield  {author} {\bibinfo {author} {\bibfnamefont {A.}~\bibnamefont
  {Ketterer}}, \bibinfo {author} {\bibfnamefont {N.}~\bibnamefont {Wyderka}},\
  and\ \bibinfo {author} {\bibfnamefont {O.}~\bibnamefont {G\"uhne}},\
  }\bibfield  {title} {\bibinfo {title} {Characterizing multipartite
  entanglement with moments of random correlations},\ }\href
  {https://doi.org/10.1103/PhysRevLett.122.120505} {\bibfield  {journal}
  {\bibinfo  {journal} {Phys. Rev. Lett.}\ }\textbf {\bibinfo {volume} {122}},\
  \bibinfo {pages} {120505} (\bibinfo {year} {2019})}\BibitemShut {NoStop}%
\bibitem [{\citenamefont {Webb}(2016)}]{Webb_2016}%
  \BibitemOpen
  \bibfield  {author} {\bibinfo {author} {\bibfnamefont {Z.}~\bibnamefont
  {Webb}},\ }\bibfield  {title} {\bibinfo {title} {The {C}lifford group forms a
  unitary 3-design},\ }\href
  {https://dl.acm.org/doi/abs/10.5555/3179439.3179447} {\bibfield  {journal}
  {\bibinfo  {journal} {Quantum Info. Comput.}\ }\textbf {\bibinfo {volume}
  {16}},\ \bibinfo {pages} {1379–1400} (\bibinfo {year} {2016})}\BibitemShut
  {NoStop}%
\bibitem [{\citenamefont {Zhu}\ \emph {et~al.}(2016)\citenamefont {Zhu},
  \citenamefont {Kueng}, \citenamefont {Grassl},\ and\ \citenamefont
  {Gross}}]{Zhu_2016}%
  \BibitemOpen
  \bibfield  {author} {\bibinfo {author} {\bibfnamefont {H.}~\bibnamefont
  {Zhu}}, \bibinfo {author} {\bibfnamefont {R.}~\bibnamefont {Kueng}}, \bibinfo
  {author} {\bibfnamefont {M.}~\bibnamefont {Grassl}},\ and\ \bibinfo {author}
  {\bibfnamefont {D.}~\bibnamefont {Gross}},\ }\bibfield  {title} {\bibinfo
  {title} {The clifford group fails gracefully to be a unitary 4-design},\
  }\href {https://arxiv.org/abs/1609.08172} {\bibfield  {journal} {\bibinfo
  {journal} {arXiv:1609.08172}\ } (\bibinfo {year} {2016})}\BibitemShut
  {NoStop}%
\bibitem [{\citenamefont {Gross}\ \emph {et~al.}(2007)\citenamefont {Gross},
  \citenamefont {Audenaert},\ and\ \citenamefont {Eisert}}]{Gross_2007}%
  \BibitemOpen
  \bibfield  {author} {\bibinfo {author} {\bibfnamefont {D.}~\bibnamefont
  {Gross}}, \bibinfo {author} {\bibfnamefont {K.}~\bibnamefont {Audenaert}},\
  and\ \bibinfo {author} {\bibfnamefont {J.}~\bibnamefont {Eisert}},\
  }\bibfield  {title} {\bibinfo {title} {Evenly distributed unitaries: On the
  structure of unitary designs},\ }\href {https://doi.org/10.1063/1.2716992}
  {\bibfield  {journal} {\bibinfo  {journal} {J. Math. Phys.}\ }\textbf
  {\bibinfo {volume} {48}},\ \bibinfo {pages} {052104} (\bibinfo {year}
  {2007})}\BibitemShut {NoStop}%
\bibitem [{\citenamefont {Collins}\ and\ \citenamefont
  {Śniady}(2006)}]{Collins_2006}%
  \BibitemOpen
  \bibfield  {author} {\bibinfo {author} {\bibfnamefont {B.}~\bibnamefont
  {Collins}}\ and\ \bibinfo {author} {\bibfnamefont {P.}~\bibnamefont
  {Śniady}},\ }\bibfield  {title} {\bibinfo {title} {Integration with respect
  to the haar measure on unitary, orthogonal and symplectic group},\ }\href
  {https://doi.org/10.1007/s00220-006-1554-3} {\bibfield  {journal} {\bibinfo
  {journal} {Commun. Math. Phys.}\ }\textbf {\bibinfo {volume} {264}},\
  \bibinfo {pages} {773} (\bibinfo {year} {2006})}\BibitemShut {NoStop}%
\bibitem [{\citenamefont {Puchała}\ and\ \citenamefont
  {Miszczak}(2017)}]{Puchala_2017}%
  \BibitemOpen
  \bibfield  {author} {\bibinfo {author} {\bibfnamefont {Z.}~\bibnamefont
  {Puchała}}\ and\ \bibinfo {author} {\bibfnamefont {J.}~\bibnamefont
  {Miszczak}},\ }\bibfield  {title} {\bibinfo {title} {Symbolic integration
  with respect to the {H}aar measure on the unitary groups},\ }\href
  {https://doi.org/10.1515/bpasts-2017-0003} {\bibfield  {journal} {\bibinfo
  {journal} {Bull. Pol. Acad. Sci.-Tech. Sci.}\ }\textbf {\bibinfo {volume}
  {65}},\ \bibinfo {pages} {21} (\bibinfo {year} {2017})}\BibitemShut {NoStop}%
\bibitem [{\citenamefont {Mele}(2024)}]{Mele_2024}%
  \BibitemOpen
  \bibfield  {author} {\bibinfo {author} {\bibfnamefont {A.~A.}\ \bibnamefont
  {Mele}},\ }\bibfield  {title} {\bibinfo {title} {Introduction to {H}aar
  {M}easure {T}ools in {Q}uantum {I}nformation: {A} {B}eginner's {T}utorial},\
  }\href {https://doi.org/10.22331/q-2024-05-08-1340} {\bibfield  {journal}
  {\bibinfo  {journal} {{Quantum}}\ }\textbf {\bibinfo {volume} {8}},\ \bibinfo
  {pages} {1340} (\bibinfo {year} {2024})}\BibitemShut {NoStop}%
\bibitem [{\citenamefont {Hill}\ and\ \citenamefont
  {Wootters}(1997)}]{Hill_1997}%
  \BibitemOpen
  \bibfield  {author} {\bibinfo {author} {\bibfnamefont {S.~A.}\ \bibnamefont
  {Hill}}\ and\ \bibinfo {author} {\bibfnamefont {W.~K.}\ \bibnamefont
  {Wootters}},\ }\bibfield  {title} {\bibinfo {title} {Entanglement of a pair
  of quantum bits},\ }\href {https://doi.org/10.1103/PhysRevLett.78.5022}
  {\bibfield  {journal} {\bibinfo  {journal} {Phys. Rev. Lett.}\ }\textbf
  {\bibinfo {volume} {78}},\ \bibinfo {pages} {5022} (\bibinfo {year}
  {1997})}\BibitemShut {NoStop}%
\bibitem [{\citenamefont {Plenio}\ and\ \citenamefont
  {Virmani}(2007)}]{Plenio_2007}%
  \BibitemOpen
  \bibfield  {author} {\bibinfo {author} {\bibfnamefont {M.~B.}\ \bibnamefont
  {Plenio}}\ and\ \bibinfo {author} {\bibfnamefont {S.}~\bibnamefont
  {Virmani}},\ }\bibfield  {title} {\bibinfo {title} {An introduction to
  entanglement measures},\ }\href
  {https://dl.acm.org/doi/abs/10.5555/2011706.2011707} {\bibfield  {journal}
  {\bibinfo  {journal} {Quantum Info. Comput.}\ }\textbf {\bibinfo {volume}
  {7}},\ \bibinfo {pages} {1–51} (\bibinfo {year} {2007})}\BibitemShut
  {NoStop}%
\end{thebibliography}
%

\newpage
\clearpage

\widetext
    
\begin{center}
\textbf{\large End Matter --- Proofs}
\end{center}

\renewcommand{\thesection}{EM\arabic{section}}

\section{States in different LU-classes with the same global moment} \label{Proof: Same moment different classes}

\begin{proof}
    Let $\ket{\phi^+_d} = \frac{1}{\sqrt{d}} \sum_{i = 0}^{d - 1} \ket{i}^{\otimes N}$, then $\mathcal{R} \qty(\ketbra{\phi^+_2}{\phi^+_2} \otimes \ketbra{\phi^+_2}{\phi^+_2}) = \mathcal{R} \qty(\ketbra{\phi^+_4}{\phi^+_4}) = 9$. These two states are neither LU-equivalent nor SLOCC-equivalent. Another example would be pure tripartite biseparable SLOCC classes, where states have the same amount of entanglement, between different parties. For example, $\ketbra{\phi^+_2}{\phi^+_2} \otimes \ketbra{0}{0}$, and $\ketbra{0}{0} \otimes \ketbra{\phi^+_2}{\phi^+_2}$.
\end{proof}

\section{Relation between marginal purities and moments} \label{Proof: Relation marginal purities and moments}

\begin{proof}
    The purity can be written as a function of $\vec{r}$: $\tr[\varrho^{(\hat{M}')2}] = \frac{1}{d^2} \sum_{i^{(\hat{M}')}, j^{(\hat{M}')} = 0}^{d^{(m)2} - 1} r_{i^{(\hat{M}')}} r_{j^{(\hat{M}')}} \tr(\lambda_{i^{(\hat{M}')}} \lambda_{j^{(\hat{M}')}}) = \frac{1}{d} \sum_{i^{(\hat{M}')}, j^{(\hat{M}')} = 0}^{d^{(m)2} - 1} r_{i^{(\hat{M}')}} r_{j^{(\hat{M}')}} \qty[\prod_{m \in \hat{M}'} \delta_{i^{(m)}, j^{(m)}}] = \frac{1}{d} \sum_{i^{(\hat{M}')} = 0}^{d^{(m)2} - 1} r_{i^{(\hat{M}')}}^2$. Other expressions can be obtained by expanding and rearranging the terms: $\sum_{i^{(\hat{M}')} = 0}^{d^{(m)2} - 1} r_{i^{(\hat{M}')}}^2 = 1 + \sum_{\hat{M} \subseteq \hat{M}'} \sum_{i^{(\hat{M})} = 1}^{d^{(m)2} - 1} r_{i^{(\hat{M})}}^2$. Notice that $\norm{\bigoplus_{\hat{M} \subseteq \hat{M}'} \vec{r}^{(\hat{M})}} = \sum_{\hat{M} \subseteq \hat{M}'} \norm{\vec{r}^{(\hat{M})}}$.
\end{proof}

\section{Relation between set of purities and moments} \label{Proof: Relation set of purities and moments}

\begin{proof}
    Lets go iteratively through the partitions in order of increasing cardinality. In each step, use equation \ref{Equation: Relation marginal purities and moments}.
    
    \textit{Step 1}: $\forall \hat{M}' \subseteq \hat{N}$ such that $\abs{\hat{M}'} = 1$, $\mathcal{P} \qty[\varrho^{(\hat{M}')}] = \mathcal{P} \qty[\tilde{\varrho}^{(\hat{M}')}] \iff 1 +  \mathcal{R} \qty[\varrho^{(\hat{M}')}] = 1 + \mathcal{R} \qty[\tilde{\varrho}^{(\hat{M}')}] \iff \mathcal{R} \qty[\varrho^{(\hat{M}')}] = \mathcal{R} \qty[\tilde{\varrho}^{(\hat{M}')}]$.
    
    \textit{Step 2}: Use the result of step 1. $\forall \hat{M}' \subseteq \hat{N}$ such that $\abs{\hat{M}'} = 2$. $\mathcal{P} \qty[\varrho^{(\hat{M}')}] = \mathcal{P} \qty[\tilde{\varrho}^{(\hat{M}')}] \iff 1 + \sum_{\substack{\hat{M}'' \subseteq \hat{M}' \\ \abs{\hat{M}''} = 1}} \mathcal{R}\qty[\varrho^{(\hat{M}'')}] + \mathcal{R}\qty[\varrho^{(\hat{M}')}] = 1 + \sum_{\substack{\hat{M}'' \subseteq \hat{M}' \\ \abs{\hat{M}''} = 1}} \mathcal{R} \qty[\tilde{\varrho}^{(\hat{M}'')}] + \mathcal{R} \qty[\tilde{\varrho}^{(\hat{M}')}] \iff \mathcal{R} \qty[\varrho^{(\hat{M}')}] = \mathcal{R} \qty[\tilde{\varrho}^{(\hat{M}')}]$. Iterations continue until $\hat{M}' = \hat{M}$.
\end{proof}

\section{Unitaries and rotations} \label{Proof: Unitaries and rotations}

\begin{proof}
    Due to the homomorphism between $\mathcal{SU}(d)$ and $\mathcal{SO}(d^2 - 1)$, every unitary transformation can be expressed by some rotation of $\vec{r}$ by some operator $Q \in \mathcal{SO}(d^2 - 1)$: $U \varrho U^\dagger = \frac{1}{d} \qty(\mathbbm{1} + \sum_{i = 1}^{d^2 - 1} r_i U \lambda_i U^\dagger ) = \frac{1}{d} \qty(\mathbbm{1} + \sum_{i, j = 1}^{d^2 - 1} q_{ji} r_i \lambda_j ) = \frac{1}{d} \qty(\mathbbm{1} + \sum_{j = 1}^{d^2 - 1} \qty[ \sum_{i = 1}^{d^2 - 1} q_{ji} r_i ] \lambda_j ) = \frac{1}{d} \qty(\mathbbm{1} + \sum_{j = 1}^{d^2 - 1} [ Q \vec{r} ]_j \lambda_j )$.
\end{proof}

\section{States with the same purity and different moments} \label{proof: Same purity different moments}

\begin{proof}
    Let $N = 2$, and $d^{(1)} = d^{(2)} = 2$. Choose two states with eigenvalue fulfilling equations $\sum_{i = 1}^4 \mu_i = 1$ and $\mathcal{P}(\varrho) = \sum_{i = 1}^4 \mu_i^2$ for some $\mathcal{P}(\varrho)$. For instance, with $\mathcal{P}(\varrho) = \frac{1}{2}$, two examples are: 1. $\mu_1 = \mu_2 = \frac{1}{2}$, $\mu_3 = \mu_4 = 0$; and 2. $\mu_1 = \mu_2 = \frac{1}{6}$, $\mu_3 = \frac{2}{3}$, $\mu_4 = 0$. Then, they are not unitary equivalent. Computing their reduced eigenvalues proves that they are neither LU-equivalent.
\end{proof}

\section{Rotations and CPTP maps} \label{Proof: Rotations and CPTP maps 2}

\begin{proof}
    If the map exists, by assumption both Bloch vectors represent valid quantum states. A well known result of linear algebra states that $\norm{\vec{r}} = \norm{\vec{\tilde{r}}}$ if and only if $\exists Q \in \mathcal{SO}(d^2 - 1)$ such that $\vec{\tilde{r}} = Q \vec{r}$.
\end{proof}

\section{Proof of theorem \ref{Theorem: LU and LO} and corollary \ref{Corollary: LU(2) and LO(3)}} \label{Proof theorem: LU and LO}

\begin{proof}
    Use the homomorphism between $\mathcal{SU}\qty[d^{(n)}]$ and $\mathcal{SO}\qty[d^{(n)2} - 1]$:
    
    $d^{(\hat{M})} \tilde{\varrho}^{(\hat{M})} - \mathbbm{1}_{\mathcal{H}} = \sum_{\hat{M}' \subseteq \hat{M}} \sum_{i^{(\hat{M}')} = 1}^{d^{(m)2} - 1} r_{i^{(\hat{M}')}} \qty[\bigotimes_{m \in \hat{M}'} U^{(m)} \lambda_{i^{(m)}} U^{(m)\dagger}] = \sum_{\hat{M}' \subseteq \hat{M}} \sum_{i^{(\hat{M}')} = 1}^{d^{(m)2} - 1} r_{i^{(\hat{M}')}} \qty[\bigotimes_{m \in \hat{M}'} \sum_{j^{(m)} = 1}^{d^{(m)2} - 1} q^{(m)}_{j^{(m)} i^{(m)}} \lambda_{j^{(m)}}] = \sum_{\hat{M}' \subseteq \hat{M}} \sum_{i^{(\hat{M}')}, j^{(\hat{M}')} = 1}^{d^{(m)2} - 1} r_{i^{(\hat{M}')}} \qty[ \prod_{m \in \hat{M}'} q^{(m)}_{j^{(m)} i^{(m)}}] \qty[\bigotimes_{m \in \hat{M}'} \lambda_{j^{(m)}}] = \sum_{\hat{M}' \subseteq \hat{M}} \sum_{j^{(\hat{M}')} = 1}^{d^{(m)2} - 1} \qty( \sum_{i^{(\hat{M}')} = 1}^{d^{(m)2} - 1} \qty[ \prod_{m \in \hat{M}'} q^{(m)}_{j^{(m)} i^{(m)}}] r_{i^{(\hat{M}')}}) \lambda_{j^{(\hat{M}')}} = \sum_{\hat{M}' \subseteq \hat{M}} \sum_{j^{(\hat{M}')} = 1}^{d^{(m)2} - 1} \qty( \sum_{i^{(\hat{M}')} = 1}^{d^{(m)2} - 1} \qty[ \bigotimes_{m \in \hat{M}'} Q^{(m)}]_{j^{(\hat{M}')} i^{(\hat{M}')}} r_{i^{(\hat{M}')}}) \lambda_{j^{(\hat{M}')}} = \sum_{\hat{M}' \subseteq \hat{M}} \sum_{j^{(\hat{M}')} = 1}^{d^{(m)2} - 1} \qty[ \bigotimes_{m \in \hat{M}'} Q^{(m)} \vec{r}]_{j^{(\hat{M}')}} \lambda_{j^{(\hat{M}')}}$.

    In terms of notation, in the second step one has to take into account that $i^{(\hat{M}')}$ is taken to have $M$ elements, with $0's$ in positions where $m \notin \hat{M}'$. In the third step it has again $M'$ elements because if $i^{(m)} = 0$, $\lambda_{i^{(m)}} = \mathbbm{1}_{\mathcal{H}^{(m)}}$ and the unitary acts trivially.
    
    In the case of qubits, i.e for the proof of corollary \ref{Corollary: LU(2) and LO(3)}, the direction $\impliedby$ is also fulfilled since the homomorphism between $\mathcal{SU}(2)$ and $\mathcal{SO}(3)$ is $2$-to-$1$. And the step from step two to step three can be argued in both directions.
\end{proof}

\section{Proof of theorem \ref{Theorem: Equivalent states have the same moments}} \label{proof theorem: Equivalent states have the same moments}

\begin{proof}
    \textit{Proof 1}: By Theorem~\ref{Theorem: LU and LO}, if $\varrho^{(\hat{M})} \sim \tilde{\varrho}^{(\hat{M})}$ then $\forall \hat{M}' \subseteq \hat{M}$, $\vec{\tilde{r}}^{(\hat{M}')} = \qty[\bigotimes_{m \in \hat{M}'} Q^{(m)}] \vec{r}^{(\hat{M}')}$, with $Q^{(m)} \in \mathcal{SO}\qty[d^{(m)^2} - 1]$, $\forall m \in \hat{M}$. This is equivalent to $\norm{\vec{\tilde{r}}^{(\hat{M}')}} = \norm{\vec{r}^{(\hat{M}')}}$, $\forall \hat{M}' \subseteq \hat{M}$.
    
    \textit{Proof 2}: If $\tilde{\varrho} = U \varrho U^\dagger$, then $\mathcal{P}(\varrho) = \mathcal{P}(\tilde{\varrho})$, and since $\varrho^{(\hat{M})} \sim \tilde{\varrho}^{(\hat{M})}$ implies $\varrho^{\qty(\hat{M}')} \sim \tilde{\varrho}^{\qty(\hat{M}')}, \quad \forall \hat{M}' \subseteq \hat{M}$; if $\varrho^{(\hat{M})} \sim \tilde{\varrho}^{(\hat{M})}$, then $\mathcal{P} \qty{\varrho^{(\hat{M})}} = \mathcal{P} \qty{\tilde{\varrho}^{(\hat{M})}}$. From $\mathcal{P} \qty{\varrho^{(\hat{M}')}} =\mathcal{P} \qty{\tilde{\varrho}^{(\hat{M}')}} \iff \mathcal{R} \qty{\varrho^{(\hat{M}')}} = \mathcal{R} \qty{\tilde{\varrho}^{(\hat{M}')}}$, this is equivalent to $\mathcal{R} \qty{\varrho^{(\hat{M})}} = \mathcal{R} \qty{\tilde{\varrho}^{(\hat{M})}}$.
\end{proof}

\section{Proof of Theorem~\ref{Theorem: Set of moments and rotations}} \label{Proof theorem: Set of moments and rotations}

\begin{proof}
    $\norm{\vec{\tilde{r}}^{(\hat{M}')}} = \norm{\vec{r}^{(\hat{M}')}}$, $\forall M' \subseteq \hat{M}$ if and only if $\exists Q^{(\hat{M}')} \in \mathcal{SO} \qty[d^{(\hat{M}')2} - 1]$, $\forall \hat{M}' \subseteq \hat{M}$ such that $\vec{\tilde{r}}^{(\hat{M}')} = Q^{(\hat{M}')} \vec{r}^{(\hat{M}')}$.
\end{proof}

\section{Proof theorem \ref{Theorem: Unitary equivalence reduced states}} \label{Proof theorem: Unitary equivalence reduced states}

\begin{proof}
    $\norm{\vec{r}^{(m)}} = \norm{\vec{\tilde{r}}^{(m)}}$ if and only if $\vec{\tilde{r}}^{(m)} = Q^{(m)} \vec{r}^{(m)}$, $Q^{(m)} \in \mathcal{SO}(3)$. This is equivalent to $\tilde{\varrho}^{(m)} = U^{(m)} \varrho^{(m)} U^{(m)\dagger}$, $U^{(m)} \in \mathcal{SU}(2)$ ($2$-to-$1$ surjective homomorphisim). Furthermore it is possible to characterize this unitaries. The latter is equivalent to  both states having the same eigenvalues: $\mu \qty{\varrho^{(m)}} = \mu \qty{\tilde{\varrho}^{(m)}}$, up to permutations of the indices. Upon properly reordering, this happens if and only if $\varrho^{(m)} = \sum_{i = 1}^2 \mu^{(m)}_i \ketbra{\psi^{(m)}_i}$ and $\tilde{\varrho}^{(m)} = \sum_{i = 1}^{2} \tilde{\mu}^{(m)}_i \ketbra{\phi^{(m)}_i}$. Finally, $\qty{\ket{\psi^{(m)}_i}}$ and $\qty{\ket{\phi^{(m)}_i}}$ are orthonormal bases of $\mathcal{H}^{(m)}$ if and only if $\exists U^{(m)} \in \mathcal{SU}(2)$ such that $\ket{\phi^{(m)}_i} = U^{(m)} \ket{\psi^{(m)}_i}$, $\forall i \in \qty{1, 2}$.
\end{proof}

\end{document}